\documentclass[]{article}

\usepackage{hyperref}
\usepackage{balance}
\usepackage{amssymb}
\usepackage{amsmath}
\usepackage{amsthm}
\usepackage[numbers]{natbib}
\newtheorem{theorem}{Theorem}
\newtheorem{proposition}[theorem]{Proposition}
\newtheorem{lemma}[theorem]{Lemma}

\allowdisplaybreaks

\begin{document}
\date{}

\title{Cobalt: BFT Governance in Open Networks}

\author{Ethan MacBrough}
\date{Ripple Research \\
      \texttt{emacbrough@ripple.com} \\ \vspace*{0.5em}
      \today}

\maketitle

\begin{abstract}
	We present Cobalt, a novel atomic broadcast algorithm that works in networks with non-uniform trust and no global agreement on participants, and is probabilistically guaranteed to make forward progress even in the presence of maximal faults and arbitrary asynchrony. The exact properties that Cobalt satisfies makes it particularly applicable to designing an efficient decentralized ``voting network" that allows a public, open-entry group of nodes to agree on changes to some shared set of rules in a fair and consistent manner while tolerating some trusted nodes and arbitrarily many untrusted nodes behaving maliciously. We also define a new set of properties which must be satisfied by any safe decentralized governance algorithm, and all of which Cobalt satisfies.
\end{abstract}

\sloppy

\section{Introduction}

With the recent explosion in popularity of decentralized digital currencies, it is becoming more imperative than ever to have algorithms that are fast, efficient, easy to run, and quantifiably safe. These digital currencies typically rely on some ``consensus" mechanism to ensure that everyone has a consistent record of which transactions occurred, to prevent malicious actors from sending the same money to two different honest actors (referred to as ``double spending"). More traditional digital currencies that rely on proof-of-work consensus \citep{nakamoto2012bitcoin}, such as Bitcoin and Ethereum, struggle to obtain low transaction times and high throughput, with theoretical results showing that proper scaling is impossible without fundamental changes to these protocols \citep{Croman2016}. Meanwhile, XRP has since its inception been both relatively fast and scalable \citep{TravisXRP}. Rejecting such proof-of-work algorithms, XRP uses a consensus algorithm in the sense of research literature \citep{schwartz2014ripple}, where a group of nodes collaborates to agree on an ordering of transactions in the face of arbitrary asynchrony and some tolerated number of arbitrarily behaving parties. It has long been known that such consensus protocols can be made very efficient \citep{Castro1999}.

For XRP the concern is thus less about how to improve the efficiency of the protocol, and more about how to enable easy ``decentralization". Traditional consensus algorithms assume a complete network where all nodes agree on who is participating in consensus. However, in a real scenario where a consensus network is run by actually independent parties with their own beliefs, regulations, and motivations, it would be effectively impossible to guarantee that everyone agrees on the same network participants. Further, trying to make such a system amenable to open participation would immediately open the door to a Sybil attack \citep{Douceur2002} wherein a single entity gains control of a substantial fraction of the network and wreaks havoc. Thus these classical consensus algorithms are a poor choice for use in a decentralized network.

The XRP Ledger Consensus Protocol (XRP LCP) resolves this issue by allowing partial disagreement on the participants in the network while still guaranteeing that all nodes come to agreement on the ledger state. The set of participants that a node considers in the network is referred to as that node's \textbf{unique node list} or \textbf{UNL}. In this setting the consistency of the network state is guaranteed by an overlap formula that prescribes a lower bound for the intersection of any two correct nodes' UNLs. As described in the original whitepaper \citep{schwartz2014ripple}, this lower bound was originally thought to be roughly $20\%$ of the UNL size. An independent response paper \citep{Armknecht2015} later suggested that the true bound was roughly $>40\%$ of the UNL size. Unfortunately, both of these bounds turned out to be naive, and in a sister paper to this paper \citep{Brad} Chase and MacBrough prove that the correct bound is actually roughly $>90\%$. Although this bound allows some variation, we would prefer a bound somewhat closer to the original expectation, to allow as much flexibility as possible. Chase and MacBrough also show that when there is not universal agreement on the participants, it is possible for the network to get ``stuck" even with $99\%$ UNL agreement and no faulty nodes, so that no forward progress can ever be made without manual intervention.

To solve these issues, this paper proposes a new consensus protocol called Cobalt, which can be used to power decentralized digital currencies such as XRP. Cobalt reduces the overlap bound to only $>60\%$, which gives much more flexibility to support painless decentralization without the fear of coming to an inconsistent ledger state. Further, unlike the previous algorithm, Cobalt cannot get stuck when the overlap bound is satisfied between every pair of honest nodes.

Another advantageous property of Cobalt is that the overlap condition for consistency is local. This means two nodes that have sufficient overlap with each other cannot arrive at inconsistent ledger states, regardless of the overlaps between other pairs of nodes. This property makes it much easier to analyze whether the network is in a safe condition. For a network that can potentially be (mis-)configured by humans, it is very important to be able to easily recognize when the network unsafe.

Further, Cobalt always makes forward progress fully asynchronously. Similar to the well-known consensus algorithm PBFT \citep{Castro1999}, the previous algorithm, XRP LCP, required assuming a form of ``weak asynchrony" where throughput could be dropped to $0$ by slightly-higher-than-expected delays or a few faulty nodes. But in practice, it is difficult to quantify what level of delay is ``expected" in a decentralized open setting, where nodes can be in arbitrary locations around the globe and have arbitrarily poor communication speed. With Cobalt however, performance simply degrades smoothly as the average message delay increases, even with the maximal number of tolerated faulty nodes and an actively adversarial network scheduler. In a live network, breaking forward progress could do a lot of damage to businesses that rely on being able to execute transactions on time, so this extra property is very valuable.

Decentralization is important primarily for two reasons: first, it gives redundancy, which protects against individual node failures and gives much higher uptime; second, it gives adaptability, so that even in the face of changing human legislation, the network can conform to those changes without needing a trusted third party that can exert singular control over the network. One of the core insights of Cobalt is that these two properties of decentralization can be separated to give better efficiency while maintaining redundancy and adaptability. Like many other decentralized consensus mechanisms, Cobalt performs relatively slowly when used as a consensus mechanism for validating transactions directly. Thus instead of using Cobalt for transactions directly, we only use it for proposing changes to the system (``amendments" in the XRP Ledger terminology). Meanwhile a separate network with universal agreement on its participants can run a faster consensus mechanism to agree on a total ordering for the transactions. Changes to the members of this ``transaction network" are executed as amendments through Cobalt. In this setup, the transaction network running a fast consensus algorithm gives both speed and redundancy, while the governance layer running Cobalt gives adaptability.

Using Cobalt together with a fast, robust transaction processing algorithm like Aardvark \citep{Clement2009} or Honeybadger \citep{Miller2016} gives all the same benefits of full decentralization while vastly improving the optimal efficiency. Further, in appendix~\ref{sectionAppendix} we present a simple protocol addition that enables the security requirements of the transaction processing algorithm to be reduced to the security requirements of Cobalt; thus even if every single transaction processing node fails, as long as the consistency requirements of Cobalt are met then every node will continue to agree on the ledger state. Other ideas for using using a decentralized algorithm to delegate a consensus group such as dBFT \citep{NEO} do not share this property, and instead require additional assumptions about the delegated group to guarantee consistency, weakening the system's overall security. The proposed addition adds only a slight latency overhead to the transaction processing algorithm.

We stress that this does not reduce the benefits of decentralization, as the transaction processing nodes only have the role of ordering transactions. Cobalt nodes still validate transactions on their own, are guaranteed to still accept the same transactions, and since client transactions are broadcast over the peer-to-peer network, the transaction processing nodes cannot even censor transactions since the Cobalt nodes could identify this behavior and eventually elect a new group of transaction processing nodes that don't censor transactions. Delegating the job of ordering transactions to a dedicated group is purely an optimization, and does not harm the robustness of the network in any way.

In section~\ref{sectionModel} we describe our network model and the problem we're trying to solve. In section~\ref{sectionOtherWork} we summarize the existing results in the area and justify the need for a new protocol. In section~\ref{sectionProtocol} we present the details of the Cobalt algorithm and prove that it satisfies all the properties we require of it. In appendix~\ref{sectionAppendix} we describe an extension that can be used to reduce the security requirements of other consensus algorithms to Cobalt's security requirements, and in appendix~\ref{sectionAppendix2} we include an extra proposition which shows that Cobalt is actually reasonably efficient, but which doesn't fit into the flow of the rest of the paper.

\section{Network Model and Problem Definition}\label{sectionModel}

Let $\mathcal{P}$ be the set of all nodes in the network. An individual node in $\mathcal{P}$ is referred to as $\mathcal{P}_i$, where $i$ is some unique identifier, such as a cryptographic public key. We do not assume all parties (or any party) know the identities of every node in $\mathcal{P}$, nor even the size of $\mathcal{P}$. We assume that every pair of nodes has a reliable authenticated communication channel between them. This can be implemented in a reasonable way by using a peer-to-peer overlay network and cryptographically signing messages. Clearly, nodes cannot be made to respond to requests from arbitrary parties, since this immediately opens up an avenue for distributed denial of service attacks \citep{Zargar2013}. We assume however that any node has some way of making requests of any every other node if it is willing to ``put in some effort". For instance, nodes might charge a modest fee or require some proof-of-work to respond to a request from an untrusted node. This makes DDOSing the network infeasible while allowing untrusted nodes to make requests of other nodes.

A node that is not crashed and behaves exactly according to the protocol defined in section~\ref{sectionProtocol} is said to be \textbf{correct}. Any node that is not correct is \textbf{Byzantine}. Byzantine behavior can include not responding to messages, sending incorrect messages, and even sending different messages to different parties. Note that in the original analysis of XRP LCP \citep{schwartz2014ripple}, it was assumed that Byzantine nodes cannot send different messages to different nodes, since it was implicitly assumed that in a peer-to-peer network such behavior would be easily identifiable. However, in our subsequent re-analysis \citep{Brad} we dispensed with this assumption, since a network partition could potentially allow irreversible damage to be done before such behavior is correctly identified. Not making this assumption is canonical in the research literature on consensus algorithms \citep{Lamport1982}, so we do not make it here either.

We further make the following nonstandard definition: a node is \textbf{actively Byzantine} if it sends some message to another node that it would not have sent had it been correct. A node can be Byzantine without being actively Byzantine; for example, a node that crashes is Byzantine but not actively Byzantine. A node which is not actively Byzantine is \textbf{honest}.

Every node $\mathcal{P}_i$ has a \textbf{unique node list} or \textbf{UNL}, denoted $\mathsf{UNL}_i$. A node's UNL is thought of as the set of nodes that it partially trusts and listens to for making decisions. $\mathsf{UNL}_i$ may or may not include $\mathcal{P}_i$ itself. The UNLs give structure to the network and allow a layered notion of trust, where a node that is present in more UNLs is implicitly considered more trustworthy and is more influential. We sometimes say that $\mathcal{P}_j$ \textbf{listens to} $\mathcal{P}_i$ if $\mathcal{P}_i\in\mathsf{UNL}_j$.

For most of the Cobalt protocol, we further assume that every honest node only has a single communication function, called \textbf{broadcast}. The statement that ``$\mathcal{P}_i$ broadcasts the message $M$" means $\mathcal{P}_i$ sends $M$ to every node that listens to $\mathcal{P}_i$. While not strictly necessary, this assumption makes the protocol analysis slightly simpler and is powerful enough on its own to develop the Cobalt protocol. The only exception to this rule is in section~\ref{sectionProtocol-Coin} for distributing threshold shares, which requires sending different messages to different nodes.

We also require that if an honest node broadcasts $\mathcal{P}_i$ a message $M$, then even if $\mathcal{P}_i$ crashes or otherwise behaves incorrectly in any way, it eventually sends $M$ to every node that listens to it, or else no node receives $M$ from $\mathcal{P}_i$. This is reasonable from an implementation standpoint if messages are routed over a peer-to-peer network: as long as a node doesn't send contradictory messages, a message sent to one party should eventually be received by all listening parties. We note that this requirement is needed only for guaranteeing liveness, not consistency.


We define the \textbf{extended UNL} $\mathsf{UNL}_i^{\infty}$ to be the ``closure" of $\mathcal{P}_i$'s UNL, which recursively contains the set of nodes in the UNL of any honest node in $\mathsf{UNL}_i^{\infty}$. Formally, this is defined inductively by defining $\mathsf{UNL}_i^{1}=\mathsf{UNL}_i$ and then defining $\mathsf{UNL}_i^{n}$ to be the set of all nodes in the UNL of any honest node in $\mathsf{UNL}_i^{n-1}$. We then define the extended UNL of $\mathcal{P}_i$ to be the set $\mathsf{UNL}_i^{\infty}=\bigcup_{n\in\mathbb{N}}\mathsf{UNL}_i^{n}$. Intuitively, a node's extended UNL represents the entire network from the perspective of $\mathcal{P}_i$; any node that could possibly have an effect on $\mathcal{P}_i$ either directly or indirectly is in $\mathsf{UNL}_i^{\infty}$.

A node $\mathcal{P}_i$ also maintains a set of \textbf{essential subsets}, denoted $\mathsf{ES}_i$, where $\mathsf{UNL}_i=\bigcup_{E\in\mathsf{ES}_i} E$. Intuitively, whereas a node's UNL is the set of all nodes that it listens to for making decisions, its essential subsets refine how it makes decisions based on the messages it receives from those nodes. The original XRP Ledger consensus algorithm had no notion of essential subsets, and instead used a predefined ``quorum" $q_i$ defining how many nodes in $\mathsf{UNL}_i$ $\mathcal{P}_i$ needs to hear from to make a decision. The direct analogue of this model would loosely be to let $\mathsf{ES}_i$ be the set of all subsets of $\mathsf{UNL}_i$ of size at least $3(n_i-q_i)+1$. It follows immediately from proposition~\ref{propDABCConsistency} that using this model with $80\%$ quorums as suggested in the XRP whitepaper, Cobalt guarantees consistency for all nodes with roughly $>60\%$ pairwise UNL overlaps.

Despite the fact that the original UNL formalism can be transferred to the essential subset model, in our model we consider the essential subsets as central and the UNL as more or less incidental. We expect a node's UNL to typically be derived automatically from its essential subsets rather than the other way around, and it is used only for bookkeeping and making some results about the algorithm easier to express.

If $S\in\mathsf{ES}_i$ for some node $\mathcal{P}_i$, we define $n_S=|S|$ and define two additional parameters, $t_S$ and $q_S$. These latter two parameters must always satisfy the following inequalities:
\begin{align}\label{eqTandQ0}
0\leqslant t_S,q_S\leqslant n_S
\end{align}
\begin{align}\label{eqTandQ1}
	t_S<2q_S-n_S.
\end{align}
\begin{align}\label{eqTandQ2}
2t_S<q_S.
\end{align}
Effectively, $t_S$ represents the maximum allowed number of actively Byzantine nodes in $S$ required for guaranteeing safety while $q_S$ represents the number of correct nodes in $S$ required for guaranteeing liveness. $q_S$ and $t_S$ can be specified by node operators individually for each $S$ as a configuration parameter; however, if two essential subsets contain the same nodes but different values of $t_S$ or $q_S$, we consider them to be distinct essential subsets. Equation~\ref{eqTandQ0} is just parameter sanity; equation~\ref{eqTandQ1} enforces that unless more than $t_S$ nodes in $S$ are actively Byzantine, then any two subsets of $q_S$ nodes must intersect in some honest node, which is used to guarantee consistency; without equation~\ref{eqTandQ2}, forward progress cannot be guaranteed to hold for any node listening to $S$ even when every single node is correct. Note that if $n_S\geqslant 3t_S+1$ and $q_S=n_S-t_S$, then all of these inequalities hold.

We make no implicit assumptions about the actual number of faulty nodes in any given essential subset $S$, nor about the total number of faulty nodes in the network. Nor do we implicitly assume any common structure to the arrangement of the essential subsets between nodes. Instead, we will explicitly show which assumptions about the allowed Byzantine nodes and the allowed essential subset configurations are needed to guarantee each result. Doing this is useful because it turns out that certain properties like consistency require much weaker assumptions than other properties like liveness. In particular, we will show that consistency is actually a ``local" property, which makes it very easy to analyze when consistency holds, and if the stronger assumptions required for liveness are ever violated, the network can at least eventually reconfigure itself to a new live configuration without having ever become inconsistent.

We call the problem we would like to solve \textbf{democratic atomic broadcast}, or \textbf{DABC}. DABC formalizes exactly the properties that are needed to implement a decentralized ``governance layer" that can be used to agree in a fair and safe way on a set of protocol rules that evolves over time.

Formally, a protocol that solves DABC allows an arbitrary (but finite) number of \textbf{proposers} -- whose identities may be unknown in advance or not universally agreed upon, and an arbitrary number of which can be Byzantine -- to broadcast \textbf{amendments} to the network. Each node can choose to either \textbf{support} or \textbf{oppose} each amendment it receives, and then each node over time \textbf{ratifies} some of those amendments and assigns each ratified amendment an \textbf{activation time}, according to the following properties:
\begin{itemize}
	\item DABC-Agreement: If any correct node ratifies an amendment $A$ and assigns it the activation time $\tau$, then eventually every other correct node also ratifies $A$ and assigns it the activation time $\tau$.
	\item DABC-Linearizability: If any correct node ratifies an amendment $A$ before ratifying some other amendment $A'$, then every other correct node ratifies $A$ before $A'$.
	\item DABC-Democracy: If any correct node ratifies an amendment $A$, then for every correct node $\mathcal{P}_i$ there exists some essential subset $S\in \mathsf{ES}_i$ such that the majority of all honest nodes in $S$ supported $A$, and further supported $A$ being ratified \textit{in the context of} all the amendments ratified before $A$.
	\item DABC-Liveness: If all correct nodes support some unratified amendment $A$, then eventually \textit{some} new amendment will be ratified.
	\item DABC-Full-Knowledge: For every time $\tau$, a correct node can run a ``waiting protocol" which always terminates in a finite amount of time, and afterwards know every amendment that will ever be ratified with an activation time before $\tau$.
\end{itemize}

We will expand on these properties in section~\ref{sectionProtocol-DABCDef} with the appropriate network conditions required for each individual property to hold. Although Agreement and Linearizability are clear and familiar from traditional atomic broadcast definitions, some explanation may be needed for the remaining three properties.

Democracy formalizes the idea that any amendment should be supported by a reasonable portion of the network. One might hope that Democracy could be strengthened to require that the majority of correct nodes in \textit{all} of $\mathcal{P}_i$'s essential subsets must have supported $A$. Unfortunately, since we don't assume universal agreement on participants, it might not be possible for a node to wait until it knows that every essential subset of every correct node has sufficient support for $A$, since there might be essential subsets that the node doesn't know about. The Democracy condition we do use seems like a reasonable compromise, and additionally it implicitly weights a node's voting power by the number of nodes that trust it. For example, if some essential subset is maintained by every single node then that subset alone could potentially pass amendments, whereas a subset only maintained by a few nodes would need to work together with other subsets to pass amendments. The stronger Democracy property \textit{does} hold in complete networks.

Most atomic broadcast algorithms use a ``Validity" or ``Censorship Resilience" property in place of Liveness that ensures a correct proposer (or client in usual terminology) will eventually have its amendment (transaction) ratified (accepted). Unfortunately, this doesn't work in our case since not every proposer may be able to broadcast its transaction to the entire network, and further an amendment might become invalid if a contradictory amendment is ratified before it. The latter issue could be solved by removing invalidated amendments \textit{post-facto}, but doing so would be unnecessarily inefficient with our protocol. Instead we use Liveness, which is equivalent to these stronger properties as long as the proposer can broadcast $A$ throughout the network and no amendments which contradict $A$ are ratified first.

For plain transaction processing, Agreement and Linearizability are the only properties needed by an atomic broadcast algorithm to guarantee consistency. Amendment processing adds a further layer of complexity though: nodes need to start acting according to the specifications of a ratified amendment at some point. Very subtle and difficult to detect bugs could surface if two nodes are running different versions of a protocol due to asynchronous knowledge of the set of ratified amendments. We rectify this issue by guaranteeing Full Knowledge, which gives nodes a way to always synchronize their active amendments. Note though that for a globally distributed network, synchronized clocks can't be assumed to exist, so each protocol built on top of a Cobalt network should first run consensus to agree on a starting time. Then every Cobalt node can agree on exactly which version of the protocol to run. This is done for example in the XRP Ledger, by agreeing on a ``ledger close time" for each block, which can be used as a starting time for the consensus protocol that agrees on the next block.

To model correctness of the algorithm, we consider a \textbf{network adversary} that is allowed to behave arbitrarily. The network adversary controls delivery of all messages as well as all Byzantine nodes. The only restrictions we make on the adversary is that it cannot break commonly accepted cryptographic protocols and eventually delivers every message sent between correct parties.

Due to the FLP result \citep{Fischer1985IDC}, a consensus algorithm (and in particular a DABC algorithm, which is a special type of consensus) cannot be guaranteed to make forward progress in the presence of arbitrary asynchrony. Thus the established convention is to ensure that consistency holds even in the presence of arbitrary asynchrony, but weaken the liveness property somehow. Two common variants are to assume liveness only holds during periods with stronger synchrony requirements \citep{Castro1999} \citep{Clement2009}, or to only make liveness hold eventually with probability $1$  \citep{BenOr1983} \citep{Bracha1984} \citep{Cachin2001} \citep{Miller2016}.

The former technique seems unsuitable for a wide-area network whose success is critical. Regardless of the heuristic likelihood of an attack breaking liveness for an extended period of time, it would be best to be mathematically confident that such an attack is infeasible. Thus we opt for the latter option for Cobalt. Although older randomness-based consensus protocols use local random values to guarantee termination, these protocols are highly inefficient in practice, requiring either exponential expected time to terminate, or asymptotically fewer tolerated faults. Newer protocols starting with \citep{Cachin2001} typically use a ``cryptographic common coin" that uses threshold signatures to generate a common random seed that cannot be predicted in advance by a computationally bounded adversary. Cryptographic common coins are very efficient, but do not immediately extend to the open network model, where the notion of a ``threshold" is undefined. We thus begin section~\ref{sectionProtocol-Coin} with defining and implementing a suitable adaptation to our model which is almost as efficient and suitably powerful to develop Cobalt.

\section{Other Work}\label{sectionOtherWork}

In complete networks where all nodes trust each other equally, there has been much research on Byzantine fault tolerant consensus algorithms, both weakly asynchronous ones and fully asynchronous ones. Notable examples include PBFT \citep{Castro1999}, SINTRA \citep{Cachin2001}, Aardvark \citep{Clement2009}, and more recently Honeybadger \citep{Miller2016}. Most of these algorithms can be made democratic using a similar democratic modification of reliable broadcast as the one presented in section~\ref{sectionProtocol-RBC}.

PBFT and Aardvark are both very fast and seem to have basic adaptations to our model, although the view change protocol requires some modification since the cryptography it uses is not fully expressive in our model (for an idea of how these changes might look, see appendix~\ref{sectionAppendix} where we develop a "view change" protocol that works in our model). However, leader-based algorithms like PBFT and Aardvark require agreement on a set of possible leaders, and if all of these leaders were to fail at once there would obviously be no way to guarantee forward progress, so these algorithms require stronger network assumptions than Cobalt. Additionally, neither of these protocols is guaranteed to make forward progress fully asynchronously, which makes them satisfy weaker properties than Cobalt. The protocol extension presented in appendix~\ref{sectionAppendix} though is loosely modeled after a simplified form of PBFT; to avoid the previously mentioned issue of needing an extra security assumption, we use Cobalt to agree on the set of possible leaders so that even if every leader fails at once eventually Cobalt can find new leaders to suggest transactions.

Meanwhile, adapting asynchronous leaderless algorithms like SINTRA and Honeybadger presents another difficulty in our model since we can't assume any specific number of honest nodes are capable of reliably broadcasting, so the reduction to asynchronous common subset used in these algorithms doesn't work. Adapting SINTRA seems especially difficult because of its significant use of threshold cryptography, for which it's not clear what an adaptation to the open model would even look like.

Alchieri et al. \citep{Alchieri2008} designed an early attempt to weaken the complete-network restrictions of classical algorithms, resulting in a Byzantine consensus algorithm that works when not all nodes know the identities of all the participants. However, in their model every node is still trusted equally, so trying to use their algorithm in an open network would immediately allow for a single entity to gain unreasonable control over the network, commonly known as a Sybil attack \citep{Douceur2002}.

Schwartz et al. developed an algorithm that works in a similar model to ours \citep{schwartz2014ripple}. It guarantees safety based on ``overlap conditions" that require that every pair of nodes trust enough nodes in common. Unfortunately, Chase and MacBrough later showed that the real safety condition is much tighter than originally thought, and further the algorithm can get stuck in certain networks where two UNLs disagree only by a single node \citep{Brad}. Further, safety is a global condition: if two nodes have sufficient overlap with each other but some other nodes don't have sufficient overlaps, then those two nodes might end up in inconsistent states anyway. This is problematic both from a usability perspective (checking safety requires checking $n^2$ overlaps rather than $n$ overlaps) and from a pragmatic perspective (my safety should not depend on the bad decisions of other nodes). Schwartz's protocol is also only weakly asynchronous, and is also not ``robust" in the sense that a small number of Byzantine nodes can prevent the protocol from ever terminating. In a live network where businesses depend on forward progress, this could be a serious problem.

More recently, Mazi\`eres described a novel protocol for solving consensus in incomplete networks \citep{Mazieres2015}. Mazi\`eres uses a network model which is similar to ours\footnote{In particular, the ``quorum slices" of Mazi\`eres's paper appear very similar to our definition of ``essential subsets". However, the way in which Mazi\`eres's algorithm uses quorum slices to determine support is different from the way Cobalt uses essential subsets: in fact, the ``quorum slices" in our model would be actually be all the sets of nodes in $\mathsf{UNL}_i$ whose intersection with every essential subset $S\in\mathsf{ES}_i$ has size at least $q_S$.} and enables very loosely-coupled network topologies to remain consistent by utilizing trust-transitivity to dynamically expand the set of nodes listened to for making decisions.

However, the concrete condition for safety is again a global condition, and seems very difficult to analyze in practice. Although the author provides a way to decide if a given Byzantine fault configuration is safe for a given topology, the condition is difficult to check in networks where each node has many quorum slices, and further there is no obvious way to input a topology and get a clear metric of how tolerant it is to Byzantine faults. This could lead to building up under-analyzed, frail topologies that seem safe but spontaneously break as soon as a single Byzantine node starts behaving dishonestly. Mazi\`eres justifies the safety of the system by comparing it to the Internet, which is a robust system that similarly takes advantage of transitive connections. In practice though, the Internet suffers transient failures due to accidental misconfigurations relatively frequently \citep{Mahajan2002}. This is not a serious problem for the Internet since it can only fail by temporarily losing connectivity; in contrast, a consensus network cannot be repaired after forking without potentially stealing money from honest actors. We therefore prefer an algorithm that is more restrictive but easier to analyze clearly; and regardless, if a node desires the greater flexibility of Mazi\`eres' protocol, then it can transitively add its peers' essential subsets out-of-protocol and get the same exact benefits. Finally, Mazi\`eres' protocol is again only weakly asynchronous and not robust.

In an attempt to resolve the inefficiency of proof-of-work, many decentralized currencies are moving towards proof-of-stake, in which a node's ``mining power" is tied to the amount of funds it locks up as collateral \citep{2017arXiv171009437B}. Although traditional proof-of-stake algorithms only guarantee asymptotic consensus and so are not applicable to our problem definition (in particular their safety depends on synchrony assumptions), another interesting avenue is to use a proof-of-stake algorithm to give nodes weighted voting power and develop a distributed consensus algorithm that is safe as long as enough of the total weighted voting power belongs to honest nodes. This idea is explored in Kwon's Tendermint protocol \citep{Kwon2014}. These protocols make decentralization easy because there is no fear of becoming inconsistent due to a misconfiguration, while avoiding Sybil attacks by tying voting power to a limited resource.

Tendermint is again not robust and requires weak asynchrony, but it seems likely that a fully asynchronous algorithm like SINTRA or Honeybadger could be adapted to this setting. However, assuming the system uses hierarchical threshold secrets in the sense proposed by Shamir \citep{Shamir1979} for instantiating common coins, then making the set of possible voting power weights even moderately fine would rapidly degrade the performance of the system, until just reconstructing a single coin value might take minutes to compute, regardless of how many participants the network has. Further, Tendermint-like protocols require listening to every node in the network, which quickly becomes inefficient in very large networks, and is only made worse when trying to adapt to full asynchrony, which typically requires $\Omega(n^3)$ messages to be exchanged to reach consensus.

Another issue is that stake in a system's success is not necessarily correlated with understanding how best to improve the system. For verifying transactions -- the use case Tendermint was designed for -- it is easy to justify tying authority to stake, since the behavior that best benefits the system is obvious and undebatable: simply run the protocol exactly as specified. For application to a governance system however, it is entirely possible for actors with good intentions to make poor decisions about how the system should operate. By allowing participants to explicitly delegate who they believe to be trustworthy, Cobalt can give authority to those who are best at making good decisions for the future of the network, rather than those who are simply incentivized against attacking the network.

Perhaps most importantly though, using proof-of-stake for determining voting power would be a poor decision for the XRP Ledger, since at the time of writing this paper, Ripple the company owns a majority of the XRP in existence, putting a dangerous amount of authority in a single location. Although Ripple is highly incentivized not to abuse this power since a loss of faith in XRP could render Ripple's XRP holdings worthless, if nothing else this gives hackers a single point of entry with which they could take over the entire network due to a careless human error.

\section{The Cobalt Protocol}\label{sectionProtocol}

In this section we describe the details of Cobalt, a protocol that solves democratic atomic broadcast in the open network model presented in section~\ref{sectionModel}. Before describing the full Cobalt protocol, we first detail certain lower level primitives that are used as part of the Cobalt algorithm. Although most of these primitives are familiar tools in the complete network model, to the author's knowledge no one else has adapted these primitives to fit our model, so we present novel instantiations of them. Since none of these protocols have been presented in our network model before, we prove by hand that every protocol is correct.

In all proofs, we make no implicit assumptions about the network connectivity or the number of Byzantine faults controlled by the adversary. If we need to assume some network connectivity or limitation on the tolerated Byzantine faults, we will state that assumption in the proposition.

Before delving into the protocols, we first develop some definitions and describe two mechanics that we use repeatedly in our protocols. These two mechanics underlie most of the basic techniques for developing consensus protocols in the complete network model, so adapting them to our model will allow us to easily adapt protocols for two of our lower level primitives, reliable broadcast and binary agreement.

Two nodes $\mathcal{P}_i$ and $\mathcal{P}_j$ are said to be \textbf{linked} if there is some essential subset $S\in\mathsf{ES}_i\cap\mathsf{ES}_j$ such that fewer than $t_S$ nodes in $S$ are actively Byzantine faulty. We say some property is \textbf{local} if the property holds between two nodes iff those two nodes are linked, regardless of whether any other nodes in the network are linked. Local properties are nice because they ensure that poorly configured nodes cannot harm correctly configured nodes. We will later prove that consistency is a local property, which we stress is very important for making the network topology easy to analyze. To the author's knowledge, Cobalt is the first incomplete network consensus algorithm for which consistency is a local property; for instance, locality does not hold for either the original XRP Ledger Consensus Protocol \citep{schwartz2014ripple} nor the protocol of Mazi\`eres \citep{Mazieres2015}.

Similarly, two nodes $\mathcal{P}_i$ and $\mathcal{P}_j$ are \textbf{fully linked} if there is some essential subset $S\in\mathsf{ES}_i\cap\mathsf{ES}_j$ such that at least $q_S$ nodes in $S$ are correct, at most $t_S$ nodes in $S$ are actively Byzantine faulty, and $t_S\leqslant n_S-q_S$. Note that if $n_S-q_S$ is greater than $t_S$, then we still allow $n_S-q_S$ nodes to be faulty, as long as they are not actively Byzantine (e.g., they can be crashed). Also note that full linkage implies linkage. While linkage is important for consistency, full linkage is important for forward progress.

A node $\mathcal{P}_i$ is \textbf{healthy} if it is honest and at most $\min\{t_S,n_S-q_S\}$ nodes in each of its essential subsets $S\in\mathsf{ES}_i$ are \textit{not} healthy. This definition can be made non-cyclical by considering a sequence of sets $F_i$ starting with $F_0$ as the set of actively Byzantine nodes and $F_i$ the set of nodes with too many $F_{i-1}$ nodes in one of its essential subsets, then taking the unhealthy nodes to be the union across the $F_i$. Healthy nodes are exactly the nodes that cannot be made to accept and/or broadcast random messages at the suggestion of actively Byzantine nodes. $\mathcal{P}_i$ is \textbf{unblocked} if it is healthy and correct, and at most $\min\{t_S,n_S-q_S\}$ nodes in each of its essential subsets $S\in\mathsf{ES}_i$ are \textit{not} unblocked. Blocked nodes can be arbitrarily prevented from terminating by the Byzantine nodes.

A node $\mathcal{P}_i$ is \textbf{strongly connected} if every pair of healthy nodes in $\mathsf{UNL}_i^{\infty}$ are fully linked with each other. Strong connectivity represents the weakest equivalent of ``global full linkage": from $\mathcal{P}_i$'s perspective, everyone in the network is fully linked. With a bit of effort, nonlocal properties can usually still be salvaged as only requiring strong connectivity rather than actually requiring that every pair of correct nodes in the network be fully linked. This is still somewhat nicer than requiring global full linkage, as at least no poorly configured nodes \textit{that you don't know about} can harm you.

The final definition we need is \textbf{weak connectivity}. A node $\mathcal{P}_i$ is \textbf{weakly connected} if $\mathcal{P}_i$ is fully linked with every healthy node in $\mathsf{UNL}_i^{\infty}$. Weak connectivity is in general much easier to guarantee than strong connectivity, since it doesn't place any requirements on how other pairs of nodes are fully linked with each other. Note though that strong connectivity only technically implies weak connectivity for healthy nodes. Generally weak connectivity is needed to guarantee that the network ``treats you fairly" and doesn't come to decisions that seem wrong to you based on what you receive from your essential subsets.

The following two lemmas provide the fundamental basis underpinning our algorithms.

\begin{lemma}\label{lemmaSupportTransfer}
	Let $\mathcal{P}_i$ be any honest node, and let $\mathcal{P}_j$ be any correct node which is fully linked with $\mathcal{P}_i$. Then if $\mathcal{P}_i$ receives some message $M$ from $q_S$ nodes in every essential subset $S\in\mathsf{ES}_i$, then eventually $\mathcal{P}_j$ will receive $M$ from $t_S+1$ nodes in some essential subset $S\in\mathsf{ES}_j$.
\end{lemma}
\begin{proof}
	Since $\mathcal{P}_i$ and $\mathcal{P}_j$ are fully linked, by definition there is some essential subset $S_{shared}\in\mathsf{ES}_i\cap\mathsf{ES}_j$. Thus if $\mathcal{P}_i$ receives some message $M$ from $q_S$ nodes in every essential subset $S\in\mathsf{ES}_i$, then in particular it receives $M$ from $q_{S_{shared}}$ nodes in $S_{shared}$. At most $t_{S_{shared}}$ of these nodes could have been actively Byzantine, so using equation~\ref{eqTandQ1},
	\begin{align*}
		q_{S_{shared}}-t_{S_{shared}}&>q_{S_{shared}}-(2q_{S_{shared}}-n_{S_{shared}}) \\
		&=n_{S_{shared}}-q_{S_{shared}} \\
		&\geqslant t_{S_{shared}},
	\end{align*}
	where the last inequality uses the definition of full linkage. Therefore at least $t_{S_{shared}}+1$ non-actively Byzantine nodes in $S_{shared}$ must have broadcast $M$. Since we assume that honest nodes can only communicate by sending the same message to everyone in that listens to them, these honest nodes must have also sent $M$ to $\mathcal{P}_j$, so eventually $\mathcal{P}_j$ will receive $M$ from $t_{S_{shared}}+1$ nodes in $S_{shared}\in\mathsf{ES}_j$.
\end{proof}

\begin{lemma}\label{lemmaSupportBlocking}
	Let $\mathcal{P}_i$ be any correct node, and let $\mathcal{P}_j$ be any correct node which is linked to $\mathcal{P}_i$. Then if $\mathcal{P}_i$ receives some message $M$ from $q_S$ nodes in every essential subset $S\in\mathsf{ES}_i$, then $\mathcal{P}_j$ cannot receive a message $M'$ that contradicts $M$ from $q_S$ nodes in every essential subset $S\in\mathsf{ES}_j$.
\end{lemma}
\begin{proof}
	By definition of linkage, there must be some $S_{shared}\in\mathsf{ES}_i\cap\mathsf{ES}_j$ such that at most $t_{S_{shared}}$ nodes in $S_{shared}$ are actively Byzantine. By the same equations as in lemma~\ref{lemmaSupportTransfer} (minus the last inequality, which requires full linkage), if $\mathcal{P}_i$ receives $M$ from $q_{S_{shared}}$ nodes in $S_{shared}$ then more than $n_{S_{shared}}-q_{S_{shared}}$ honest nodes in $S_{shared}$ sent $M$. Since honest nodes cannot broadcast both $M'$ and $M$, fewer than $n_{S_{shared}}-(n_{S_{shared}}-q_{S_{shared}})=q_{S_{shared}}$ nodes in $S_{shared}$ can send $M'$ to $\mathcal{P}_j$.
\end{proof}

In light of the previous lemmas, we make two more definitions. A node $\mathcal{P}_i$ sees \textbf{strong support} for a message $M$ if $\mathcal{P}_i$ receives $M$ from $q_S$ nodes in every essential subset $S\in\mathsf{ES}_i$. Similarly, $\mathcal{P}_i$ sees \textbf{weak support} for a message $M$ if $\mathcal{P}_i$ receives $M$ from $t_S+1$ nodes in some essential subset $S\in\mathsf{ES}_i$.

Using these definitions, lemma~\ref{lemmaSupportTransfer} can be phrased as ``fully linked nodes have enough overlap to where if one node sees strong support then the other will eventually see weak support", and lemma~\ref{lemmaSupportBlocking} can be phrased as ``linked nodes have enough overlap to where they cannot simultaneously both see strong support for contradictory messages". It turns out that relating nodes in these two ways is enough to recover most of the techniques used in developing BFT algorithms from the complete network case, allowing us to easily adapt many algorithms to our model.

\subsection{Cryptographic Randomness}\label{sectionProtocol-Coin}

Before we can define the Cobalt protocol, one remaining piece needs to be developed. As mentioned at the end of section~\ref{sectionModel}, Cobalt uses cryptography to generate common pseudorandom values that are unpredictable by the network adversary in order to sidestep the FLP result \citep{Fischer1985IDC}.


Let $\mathcal{S}$ be a probability space with probability measure $P$. We define a \textbf{common random source} or \textbf{CRS} to be a protocol where nodes can \textbf{sample} at any time, and then output some value according to the following properties:
\begin{itemize}
	\item CRS-Consistency: If any honest node outputs $s$, then no honest node linked to it ever outputs $s'\neq s$.
	\item CRS-Termination: If $\mathcal{P}_i$ is strongly connected and every unblocked node in $\mathsf{UNL}_i^{\infty}$ samples the CRS, then every unblocked node in $\mathsf{UNL}_i^{\infty}$ eventually produces an output.
	\item CRS-Randomness: Suppose $\mathcal{P}_i$ is correct and weakly connected, at most $t_S$ nodes in every essential subset $S\in\mathsf{ES}_i$ are controlled by the adversary, and $\mathcal{P}_i$ eventually outputs $s$. Then for any value $x$ produced by the adversary before any healthy node in $\mathsf{UNL}_i$ has sampled the CRS, with overwhelming probability $|\mathrm{Pr}[s=x]-P(x)|\leqslant\epsilon$ for negligible $\epsilon$.
\end{itemize}

The last property formalizes the idea that the adversary cannot get a significantly better prediction of the random output than it would by just randomly picking a value from $\mathcal{S}$.

We postpone describing the concrete details of this protocol until appendix~\ref{sectionAppendix3}.

\subsection{Reliable Broadcast}

\subsubsection{Definition}

\textbf{Reliable broadcast}, or RBC, is a basic primitive that allows a specified \textbf{broadcaster} to send a message to the network, and guarantees that even if the broadcaster is Byzantine faulty, it must send the same message to every node. For the protocol definition, the broadcaster may or may not be a node within the network; however, when using RBC within Cobalt we only ever use it in the context where the broadcaster \textit{is} a node in the network.

More formally, a reliable broadcast protocol is any protocol where a specified broadcaster entity $\mathcal{B}_i$ inputs an arbitrary message, and every node can \textbf{accept} some message, subject to the following properties:
\begin{itemize}
	\item RBC-Consistency: If any honest node accepts a message $M$, then no honest node linked to it ever accepts any message $M'\neq M$.
	\item RBC-Reliability: If $\mathcal{P}_i$ is strongly connected and any healthy node in $\mathsf{UNL}_i^{\infty}$ accepts a message $M$, then every unblocked node in $\mathsf{UNL}_i^{\infty}$ eventually accepts $M$.
	\item RBC-Validity: If $\mathcal{B}_i$ is honest and inputs the message $M$, then any healthy node that accepts a message must accept $M$.
	\item RBC-Non-Triviality: If $\mathcal{B}_i$ is honest \textit{and} can broadcast to every correct node in the network, then eventually every unblocked node will accept $M$.
\end{itemize}
Most researchers combine Consistency and Reliability into one property, but we keep them separate since the network assumptions needed for Consistency are so much weaker. Most researchers also combine Validity and Non-Triviality, since its assumed that every node can broadcast to the entire network. Since in our network model we do not assume that all nodes have communication channels between them, $\mathcal{B}_i$ might be isolated from the rest of the network, so combining these properties doesn't work.

\subsubsection{Protocol}\label{sectionProtocol-RBC}

In the complete network model, the canonical reliable broadcast protocol is due to Bracha \citep{Bracha1984}. Our protocol is closely modeled after Bracha's protocol, and behaves exactly the same in the complete network case.

The protocol begins by having $\mathcal{B}_i$ broadcast $INIT(M)$ to everyone listening to it. After that, each node $\mathcal{P}_j$ (including $j=i$, if $\mathcal{B}_i$ is a member of the network) runs the following protocol\footnote{In our protocol descriptions, we use the underscore notation $\mathunderscore$ to refer to ``any possible value".}.
\begin{enumerate}
	\item\label{RBCstep1} Upon receiving an $INIT(M)$ message directly from $\mathcal{B}_i$, broadcast $ECHO(M)$ if we have not yet broadcast $ECHO(\mathunderscore)$.
	\item\label{RBCstep2} Upon receiving weak support for $ECHO(M)$, broadcast $ECHO(M)$ if we have not yet broadcast $ECHO(\mathunderscore)$.
	\item\label{RBCstep3} Upon receiving strong support for $ECHO(M)$, broadcast $READY(M)$ if we have not yet broadcast $READY(\mathunderscore)$.
	\item\label{RBCstep4} Upon receiving weak support for $READY(M)$, broadcast $READY(M)$ if we have not yet broadcast $READY(\mathunderscore)$.
	\item\label{RBCstep5} Upon receiving strong support for $READY(M)$, accept $M$.
\end{enumerate}

When multiple instances of reliable broadcast might be running at the same time, we tag each message with a unique instance id to differentiate them.

Step~\ref{RBCstep2} is not technically necessary, but it makes it somewhat easier to reliably broadcast to the network. Note that since we assume that every message is cryptographically signed by the sender, if we also include the public key of $\mathcal{B}_i$ (which may not be known to all nodes) in the instance tag, then in step~\ref{RBCstep1} we could actually broadcast $ECHO(M)$ even if we only receive $ECHO(M)$ from a single node, as long as we also include $\mathcal{B}_i$'s signature with it. This would make it even easier for nodes to reliably broadcast to the network. The only security risk for allowing more nodes to reliably broadcast is the possibility of allowing spam to congest the network; since spammers can be eventually excluded, there is little value in trying to make it harder for nodes to reliably broadcast.

\subsubsection{Analysis}\label{sectionProofs-RBC}

Reliable broadcast can be split into two phases: the ``echo" phase and the ``ready" phase, distinguished by the labels on the messages from each phase. Roughly speaking, the echo phase serves to guarantee that everyone accepts the same message (consistency) while the second phase guarantees that if anyone accepts a message then so does everyone else (reliability).

\begin{proposition}\label{propRBCConsistency}
	Suppose two correct nodes $\mathcal{P}_i$ and $\mathcal{P}_j$ are linked and they accept the messages $M$ and $M'$, respectively. Then $M=M'$.
\end{proposition}
\begin{proof}
	By step~\ref{RBCstep5} of the RBC algorithm, a node only accepts a message $M$ if it receives $READY(M)$ strong support for $M$. Since RBC restricts nodes to only broadcast a single message for each label, by lemma~\ref{lemmaSupportBlocking}, $M=M'$.
\end{proof}

Although consistency is local as the previous proposition shows, unfortunately the stronger property of reliability is not local.

\begin{lemma}\label{propRBCReadyBlocking}
	Suppose $\mathcal{P}_k$ is strongly connected and two healthy nodes $\mathcal{P}_i,\mathcal{P}_j\in\mathsf{UNL}_k^{\infty}$ broadcast $READY(M)$ and $READY(M')$, respectively. Then $M=M'$.
\end{lemma}
\begin{proof}
	By steps~\ref{RBCstep3} and \ref{RBCstep4} of the reliable broadcast protocol, an honest node $\mathcal{P}_i$ can only broadcast $READY(M)$ for some message $M$ if either $1)$ it received strong support for $ECHO(M)$, or $2)$ it received weak support for $READY(M)$. In the latter case, if $\mathcal{P}_i$ is healthy then this implies in particular that some healthy node in $\mathsf{UNL}_i\subseteq\mathsf{UNL}_k^{\infty}$ broadcast $READY(M)$ \textit{before} $\mathcal{P}_i$. Since there are only a finite number of nodes in $\mathsf{UNL}_k^{\infty}$, there must exist some healthy node $\mathcal{P}_{i'}$ in $\mathsf{UNL}_k^{\infty}$ that broadcast $READY(M)$ before any other healthy node in its UNL. In particular, $\mathcal{P}_{i'}$ must have broadcast $READY(M)$ due to having received strong support for $ECHO(M)$.
	
	Thus if two healthy nodes $\mathcal{P}_i,\mathcal{P}_j\in\mathsf{UNL}_k^{\infty}$ broadcast $READY(M)$ and $READY(M')$, respectively, then we can assume that there are two healthy nodes $\mathcal{P}_{i'},\mathcal{P}_{j'}\in\mathsf{UNL}_k^{\infty}$ such that $\mathcal{P}_{i'}$ received strong support for $ECHO(M)$ while $\mathcal{P}_{j'}$ received strong support for $ECHO(M')$. Since $\mathcal{P}_k$ is strongly connected by assumption, $\mathcal{P}_{i'}$ and $\mathcal{P}_{j'}$ are linked, so by lemma~\ref{lemmaSupportBlocking} $M=M'$.
\end{proof}

\begin{proposition}\label{propRBCReliability}
	If $\mathcal{P}_k$ is strongly connected and any healthy node $\mathcal{P}_i\in\mathsf{UNL}_k^{\infty}$ accepts the message $M$, then every unblocked node $\mathcal{P}_j\in\mathsf{UNL}_k^{\infty}$ will eventually accept $M$.
\end{proposition}
\begin{proof}
	Since every pair of healthy nodes in $\mathsf{UNL}_k^{\infty}$ are fully linked by assumption, if $\mathcal{P}_i$ accepts $M$ then by lemma~\ref{lemmaSupportTransfer}, eventually every unblocked node in $\mathsf{UNL}_k^{\infty}$ will eventually see weak support for $READY(M)$. By lemma~\ref{propRBCReadyBlocking}, no healthy node in $\mathsf{UNL}_k^{\infty}$ can have previously broadcast $READY(M')$ for any $M'\neq M$, so by step~\ref{RBCstep4} of the RBC protocol, eventually every healthy and correct node in $\mathsf{UNL}_k^{\infty}$ broadcasts $READY(M)$. In particular, if $\mathcal{P}_j\in\mathsf{UNL}_k^{\infty}$, then every healthy and correct node in $\mathsf{UNL}_j\subseteq\mathsf{UNL}_k^{\infty}$ eventually broadcasts $READY(M)$, so if $\mathcal{P}_j$ is unblocked then eventually $\mathcal{P}_j$ receives strong support for $READY(M)$. Thus $\mathcal{P}_j$ accepts $M$ by step~\ref{RBCstep5} of the protocol.
\end{proof}

\begin{proposition}\label{propRBCValidity}
	If $\mathcal{B}_i$ is honest, then no healthy node can accept a message not broadcast by $\mathcal{B}_i$.
\end{proposition}
\begin{proof}
	This follows from a simple analysis of the protocol by noting that a healthy node can't broadcast $ECHO(M)$ without either receiving $INIT(M)$ from $\mathcal{B}_i$ or receiving $ECHO(M)$ from another healthy node. Thus if $\mathcal{B}_i$ only broadcasts $INIT(M)$, then no healthy node will broadcast $ECHO(M')$ for any $M'\neq M$. By similar logic, no healthy node will broadcast $READY(M')$ for any $M'\neq M$, so no healthy node will ever see enough $READY(M')$ messages to accept $M'$.
\end{proof}

\begin{proposition}\label{propRBCNonTriviality}
	If $\mathcal{B}_i$ is correct and can broadcast to every correct node in the network, then eventually every unblocked node will accept $M$.
\end{proposition}
\begin{proof}
	Since every node can receive $INIT(M)$ from $\mathcal{B}_i$, every healthy and correct node will broadcast $ECHO(M)$, so eventually every healthy and correct node will broadcast $READY(M)$, so eventually every unblocked node will accept $M$.
\end{proof}

\begin{theorem}\label{thmRBC}
	The RBC protocol defined in section~\ref{sectionProtocol-RBC} satisfies the properties of a reliable broadcast algorithm in the open network model.
\end{theorem}
\begin{proof}
	Consistency is proven in proposition~\ref{propRBCConsistency}. Reliability is proven in proposition~\ref{propRBCReliability}. Validity is proven in proposition~\ref{propRBCValidity}. Non-triviality is proven in proposition~\ref{propRBCNonTriviality}.
\end{proof}

\subsubsection{Democratic Reliable Broadcast}

We will also find useful a slight variation on RBC called \textbf{democratic reliable broadcast} or \textbf{DRBC}.

A DRBC protocol is similar to RBC except it allows nodes to choose whether to \textbf{support} or \textbf{oppose} messages that are broadcast, and replaces non-triviality with the following properties:
\begin{itemize}
	\item DRBC-Democracy: If any healthy node $\mathcal{P}_i$ is weakly connected and accepts the message $M$, then there exists some essential subset $S\in \mathsf{ES}_i$ such that the majority of all honest nodes in $S$ supported $M$.
	\item DRBC-Censorship-Resilience: If a $\mathcal{B}_i$ can broadcast to every correct node in the network, and all correct nodes support $M$, then eventually every unblocked node will accept $M$.
\end{itemize}

One can easily transform the above RBC protocol into a DRBC protocol by specifying that each node only broadcasts an $ECHO(M)$ message iff it supports $M$ (note though that a node may still need to broadcast $READY(M)$ even if it doesn't support $M$).

\begin{proposition}\label{propDRBCDemocracy}
	If any healthy node $\mathcal{P}_k$ is weakly connected and accepts the message $M$, then there is some essential subset $S\in\mathsf{ES}_k$ such that the majority of honest nodes in $S$ supported $M$.
\end{proposition}
\begin{proof}
	If any healthy node in $\mathsf{UNL}_k^{\infty}$ broadcasts $READY(M)$, there must have been a healthy node $\mathcal{P}_i\in\mathsf{UNL}_k^{\infty}$ that was the first healthy node in $\mathsf{UNL}_k^{\infty}$ to broadcast $READY(M)$. Then $\mathcal{P}_i$ must have seen strong support for $ECHO(M)$. By weak connectivity, $\mathcal{P}_i$ and $\mathcal{P}_k$ are fully linked (and in particular, linked), so there must be some essential subset $S\in\mathsf{ES}_k$ such that at least $q_S-t_S$ honest nodes in $S$ broadcast $ECHO(M)$, while at most $n_S-q_S$ honest nodes in $S$ did not broadcast $ECHO(M)$. By equation~\ref{eqTandQ1}, $q_S-t_S>q_S-(2q_S-n_S)=n_S-q_S$, so the majority of honest nodes in $S$ must have supported $M$.
\end{proof}

\begin{theorem}
	The modified protocol defined in section~\ref{sectionProtocol-RBC} satisfies the properties of a democratic reliable broadcast algorithm in the open network model.
\end{theorem}
\begin{proof}
	Consistency, reliability, and validity all still hold with the modified algorithm, since none of the proofs for those properties in theorem~\ref{thmRBC} assume that any nodes are guaranteed to broadcast an $ECHO$ message. Democracy is proven in proposition~\ref{propDRBCDemocracy}.
	
	The proof of Censorship Resilience is identical to the proof of RBC-Non-Triviality, since if every correct node supports $M$ then eventually every healthy and correct node will broadcast $ECHO(M)$.
\end{proof}

\subsection{Binary Agreement}

\subsubsection{Definition}

The other low level primitive we need is \textbf{asynchronous binary Byzantine agreement} or \textbf{ABBA}. ABBA is the most basic consensus primitive: every node inputs some bit, and then all the nodes agree on a single bit that was input by some honest node.

More formally, an ABBA protocol allow each node to input a single bit, and then every node outputs a single bit according to the following properties:
\begin{itemize}
	\item ABBA-Consistency: Two honest, linked nodes cannot output different values.
	\item ABBA-Termination: If $\mathcal{P}_k$ is strongly connected and every unblocked node in $\mathsf{UNL}_k^{\infty}$ provides some input to the algorithm, then eventually every unblocked node in $\mathsf{UNL}_k^{\infty}$ terminates with probability $1$.
	\item ABBA-Validity: If any unblocked node outputs $v$, then some unblocked node must have input $v$.
\end{itemize}

The above definition of Validity is common in the complete network model, but it turns out to be too weak for our purposes. Indeed, an algorithm that only satisfies the above Validity property could decide $1$ even if some totally isolated honest node were the only node that voted $1$. We thus actually need a stronger notion of validity to guarantee correctness of Cobalt:
\begin{itemize}
	\item ABBA-Strong-Validity: If any unblocked node $\mathcal{P}_i$ outputs $v$, then there is some chain of unblocked nodes $\mathcal{P}_i=\mathcal{P}_{i_0},\mathcal{P}_{i_1},...,\mathcal{P}_{i_n}$, where for all $k\leqslant n$, $\mathcal{P}_{i_{k}}\in\mathsf{UNL}_{i_{k-1}}$, and the node $\mathcal{P}_{i_n}$ input $v$.
\end{itemize}

Although rather awkward, the Strong Validity property turns out to be just strong enough for our purposes.

\subsubsection{Protocol}\label{sectionProtocol-ABBA}

Our ABBA protocol is based off of a binary agreement protocol designed for complete networks by Most\'efaoui et al. \citep{Mostefaoui2014}. The protocol by Most\'efaoui et al. is fully asynchronous and uses a CRS in the form of a ``common coin". It takes longer on average to terminate compared to an earlier protocol in the same model developed by Cachin et al. \citep{Cachin2005}; unfortunately it seems impossible to develop a simple adaptation for Cachin et al.'s protocol, since the cryptographic proofs it uses to justify messages don't seem to work in our model\footnote{Of course, threshold signatures as used in Cachin et al.'s original specification don't work in our model. But even replacing threshold signatures with multisignatures, if a node $\mathcal{P}_i$ broadcasts a ``main message" voting $1$ after seeing $q_S$ valid ``pre messages" voting $1$ from every $S\in\mathsf{ES}_i$, then because not all nodes know each other's essential subsets, the validity proof of this main message only proves to $\mathcal{P}_j$ that \textit{some} $S\in\mathsf{ES}_j$ sent $q_S$ valid pre messages voting $1$ to $\mathcal{P}_i$; but $\mathcal{P}_j$ then still doesn't know if there might be some node $\mathcal{P}_k$ for which \textit{no} $S\in\mathsf{ES}_k$ sent $q_S$ valid pre messages voting $1$ to $\mathcal{P}_i$. Thus a Byzantine node could send opposite valid main messages to two nodes that don't know about each other, and guarantee that they never agree.}

For the protocol, we use a sequence $\rho_r$ of common random sources that each sample uniformly from $\{0,1\}$ for every $r\geqslant 0$.

The protocol works as follows, run from the perspective of $\mathcal{P}_i$:
\begin{enumerate}
	\item\label{ABBAfinish1} Upon receiving weak support for $FINISH(v)$ for some binary value $v$, broadcast $FINISH(v)$ if we haven't yet broadcast $FINISH(\mathunderscore)$.
	\item\label{ABBAfinish2} Upon receiving strong support for $FINISH(v)$, output $v$ and terminate.
	\item\label{ABBAsetup} Set $\mathsf{values}_i^r=\emptyset$ for all $r\geqslant 0$. Upon $\mathcal{P}_i$ providing an input value $v_{in}$, set $r=0$ and $\mathsf{est}_i^r=v_{in}$.
	\item\label{ABBAbeginLoop} Broadcast $INIT(\mathsf{est}_i^r, r)$.
	\item\label{ABBAinitReliability} Upon receiving weak support for $INIT(v,r)$, broadcast $INIT(v,r)$.
	\item\label{ABBAaddToValues} Upon receiving strong support for $INIT(v,r)$, add $v$ to $\mathsf{values}_i^r$ and broadcast $AUX(v,r)$ if we have not already broadcast $AUX(\mathunderscore, r)$.
	\item\label{ABBAaux} For every essential subset $S\in\mathsf{ES}_i$, wait until there exists some subset $T\subseteq S$, such that $\vert T\vert\geqslant q_S$ and from every node in $T$ we received $AUX(v,r)$ for some $v\in\mathsf{values}_i^r$ (possibly different $v$ for different nodes in $T$). Then broadcast $CONF(\mathsf{values}_i^r,r)$.
	\item\label{ABBAconfStep} For every essential subset $S\in\mathsf{ES}_i$, wait until there exists some subset $T\subseteq S$, such that $\vert T\vert\geqslant q_S$ and from every node in $T$ we received $CONF(C,r)$ for some $C\subseteq\mathsf{values}_i^r$ (possibly different $C$ for different nodes in $T$).
	\item\label{ABBArandom} Sample from $\rho_r$ and place its value in $s_r$.
	\item\label{ABBAfinal} If $|\mathsf{values}_i^r|=2$, then set $\mathsf{est}_i^{r+1}=s_r$. If $\mathsf{values}_i^r=\{v\}$ for some $v$, then set $\mathsf{est}_i^{r+1}=v$. If in fact $\mathsf{values}_i^r=\{s_r\}$, then additionally broadcast $FINISH(s_r)$ if we have not yet broadcast $FINISH(\mathunderscore)$.
	
	Set $r=r+1$ and return to step~\ref{ABBAbeginLoop}.
\end{enumerate}

The above protocol is defined asynchronously, so that once you get to some step in the protocol you keep running that step forever if its logic has not been satisfied by the time you get to the next step. So for instance, the logic involving the $FINISH$ messages in steps~\ref{ABBAfinish1} and~\ref{ABBAfinish2} should be continuously checked even after you get to the later steps.

The original protocol of Most\'efaoui et al. did not use the $CONF$ messages or the $FINISH$ messages. The $FINISH$ messages are necessary for guaranteeing consistency is a local property. The $CONF$ messages are necessary because our definition of a CRS is weaker than a true common coin as assumed in the original protocol. The use of $CONF$ messages in step~\ref{ABBAconfStep} ensures that if any node $\mathcal{P}_i$ gets to step~\ref{ABBAfinal} with $\mathsf{values}_i^r=\{v\}$, then the value of $s_r$ is practically independent of the value of $v$.

\subsubsection{Analysis}\label{sectionProofs-ABBA}

\begin{proposition}\label{propABBAAgreement}
	If two honest nodes $\mathcal{P}_i$ and $\mathcal{P}_j$ are linked, then they cannot output different binary values.
\end{proposition}
\begin{proof}
	Since an honest node can only broadcast a single $FINISH$ message, by the condition for outputting a binary value $v$ in step~\ref{ABBAfinish2} and lemma~\ref{lemmaSupportBlocking}, $\mathcal{P}_i$ and $\mathcal{P}_j$ cannot output different values.
\end{proof}

The above proposition shows why we use the $FINISH$ message. Note that the part of the protocol involving the $FINISH$ message is not present in Most\'efaoui et al.'s algorithm. The original version instead has nodes that get $\mathsf{values}_r=\{s_r\}$ for some round $r$ wait until they sample some CRS $\rho_{r'}$ with $r'>r$ that returns $s_{r'}=s_r$. This change is not fundamental to the open network model (indeed, the original version works fine in our model, and our version works fine in Most\'efaoui et al.'s model). However, as shown in \ref{propABBAAgreement}, adding the $FINISH$ message makes agreement a local property, which is a great bonus in the open network model. Thus we prefer the modified version, even though it incurs an extra communication round. Without using the $FINISH$ message step, the above proposition does not hold, since nodes can realize ABBA has terminated in different rounds, and unlinked nodes in a late terminator's UNL can shift their opinions to the opposite value after the earlier node has already terminated.

We now move onto proving termination and validity. These properties are significantly more involved than agreement, so we try to break the proofs into the smallest chunks possible.

Each round of the binary agreement protocol described in section~\ref{sectionProtocol-ABBA} breaks roughly into three phases. Similar to the case of RBC, the phases can be divided by the labels on the messages involved in each phase: the first phase is the ``initialization" phase, and comprises steps~\ref{ABBAinitReliability} and \ref{ABBAaddToValues} involving the $INIT$ messages; the second phase is the ``auxiliary" phase in steps~\ref{ABBAaddToValues} and \ref{ABBAaux} that involves the $AUX$ messages; the third phase is the ``confirmation" phase in steps~\ref{ABBAaddToValues} and \ref{ABBAconfStep} that involves the $CONF$ messages.

We begin by proving lemmas representing the correctness of the initialization phase.

\begin{lemma}\label{lemmaABBAValuesValidity}
	If $\mathcal{P}_i$ is unblocked and adds $v$ to $\mathsf{values}_i^r$, then there is some chain of unblocked nodes $\mathcal{P}_i=\mathcal{P}_{i_0},\mathcal{P}_{i_1},...,\mathcal{P}_{i_n}$, where for all $k\leqslant n$, $\mathcal{P}_{i_{k}}\in\mathsf{UNL}_{i_{k-1}}$, and $\mathsf{est}_{i_n}^r=v$.
\end{lemma}
\begin{proof}
	If $\mathcal{P}_i$ adds $v$ to $\mathsf{values}_i^r$, then certainly some unblocked node $\mathcal{P}_{i_1}\in\mathsf{UNL}_i$ must have broadcast $INIT(r,v)$ by the logic in step~\ref{ABBAaddToValues} for adding a value to $\mathsf{values}_i^r$. But an unblocked node $\mathcal{P}_{i_k}$ only broadcasts $INIT(r,v)$ if either $\mathsf{est}_{i_k}^r=v$ or there was some unblocked node in \textit{its} UNL that broadcast $INIT(r,v)$ before $\mathcal{P}_{i_k}$ did. By repeating, we successively build up the chain of unblocked nodes until we eventually reach some unblocked node that had $\mathsf{est}_{i_n}^r=v$, since $\mathsf{UNL}_i^{\infty}$ is finite implying that at some point we must reach an unblocked node that sent $INIT(r,v)$ before any other unblocked node in its UNL.
\end{proof}

\begin{lemma}\label{lemmaABBAValuesAgreement}
	If $\mathcal{P}_k$ is strongly connected and any honest node $\mathcal{P}_i\in\mathsf{UNL}_k^{\infty}$ adds $v$ to $\mathsf{values}_i^r$, then every unblocked node $\mathcal{P}_j\in\mathsf{UNL}_k^{\infty}$ will eventually add $v$ to $\mathsf{values}_j^r$.
\end{lemma}
\begin{proof}
	Identical to the proof of proposition~\ref{propRBCReliability}.
\end{proof}

\begin{lemma}\label{lemmaABBAValuesTermination}
	If $\mathcal{P}_k$ is strongly connected, every unblocked node in $\mathsf{UNL}_k^{\infty}$ gets to step~\ref{ABBAbeginLoop} for round $r$, and no unblocked nodes in $\mathsf{UNL}_k^{\infty}$ terminate in round $r$, then eventually every unblocked node $\mathcal{P}_j\in\mathsf{UNL}_k^{\infty}$ adds some value to $\mathsf{values}_j^r$.
\end{lemma}
\begin{proof}
	For convenience, given some essential subset $S$ define the \textbf{majority input} $v_S$ to be the binary value set for $\mathsf{est}_i^r$ by the majority of unblocked nodes $\mathcal{P}_i\in S$. Then once all these unblocked nodes get to step~\ref{ABBAbeginLoop} in round $r$, if any unblocked node $\mathcal{P}_i$ listens to $S$ there must be at least $q_S$ unblocked nodes in $S$, so $\mathcal{P}_i$ will eventually receive $INIT(r,v_S)$ messages from more than $q_S/2>t_S$ nodes in $S$, causing $\mathcal{P}_i$ to broadcast $INIT(r,v_S)$ according to the condition in step~\ref{ABBAinitReliability}.
	
	Let $\mathcal{P}_i\in\mathsf{UNL}_k^{\infty}$ be some unblocked node. Suppose every essential subset $S\in\mathsf{ES}_i$ has the same majority vote $v$. Then since $\mathsf{P}_i\subseteq\mathsf{UNL}_k^{\infty}$, $\mathcal{P}_i$ is fully linked with every unblocked node in $\mathsf{UNL}_i$, so eventually every unblocked node in $\mathsf{UNL}_i$ broadcasts $INIT(r,v)$ by the preceding paragraph. Thus $\mathcal{P}_i$ adds $v$ to $\mathsf{values}_i^r$ in step~\ref{ABBAaddToValues}, and by lemma~\ref{lemmaABBAValuesAgreement} every node $\mathsf{P}_j\in\mathsf{UNL}_k^{\infty}$ also eventually adds $v$ to $\mathsf{values}_j^r$.
	
	It remains to show the case where every unblocked node in $\mathsf{UNL}_k^{\infty}$ maintains two essential subsets $S,S'$ with $v_S\neq v_{S'}$. But in this case by the first paragraph every unblocked node in $\mathsf{UNL}_k^{\infty}$ eventually broadcasts both $INIT(r,0)$ and $INIT(r,1)$. Thus every unblocked node $\mathsf{P}_j\in\mathsf{UNL}_k^{\infty}$ eventually adds \textit{both} $0$ and $1$ to $\mathsf{values}_j^r$.
\end{proof}

Note that in the previous lemma the reason why we needed to specify ``no unblocked nodes in $\mathsf{UNL}_k^{\infty}$ terminate in round $r$" is because a node can possibly terminate at any time if it receives enough $FINISH$ messages, and therefore stop participating before adding a value to $\mathsf{values}^r$.

We now move onto the auxiliary phase.

\begin{lemma}\label{lemmaABBAAuxAgreement}
	If two honest nodes $\mathcal{P}_i$ and $\mathcal{P}_j$ are linked, then if $\mathcal{P}_i$ continues to step~\ref{ABBAconfStep} in round $r$ with $\mathsf{values}_i^r=\{v\}$, $\mathcal{P}_j$ cannot continue to step~\ref{ABBAconfStep} in round $r$ with $\mathsf{values}_j^r=\{\neg v\}$.
\end{lemma}
\begin{proof}
	In order to progress to step~\ref{ABBAconfStep} with $\mathsf{values}_i^r=\{v\}$, $\mathcal{P}_i$ must receive strong support for $AUX(v,r)$. The lemma thus holds immediately by lemma~\ref{lemmaSupportBlocking}.
\end{proof}

Note that the above proposition doesn't guarantee that $\mathcal{P}_j$ will continue to step~\ref{ABBAconfStep} with $\mathsf{values}_j^r=\{v\}$. Instead $\mathcal{P}_j$ might continue to step~\ref{ABBAconfStep} with $\mathsf{values}_j^r=\{0,1\}$.

\begin{lemma}\label{lemmaABBAAuxTermination}
	If $\mathcal{P}_k$ is strongly connected, every unblocked node in $\mathsf{UNL}_k^{\infty}$ gets to step~\ref{ABBAbeginLoop} for round $r$, and no unblocked nodes in $\mathsf{UNL}_k^{\infty}$ terminate in round $r$, then eventually every unblocked node in $\mathsf{UNL}_k^{\infty}$ either progresses to step~\ref{ABBAconfStep} in round $r$ or terminates.
\end{lemma}
\begin{proof}
	By lemma~\ref{lemmaABBAValuesTermination}, eventually every unblocked node in $\mathsf{UNL}_k^{\infty}$ broadcasts an $AUX$ message in round $r$. Further, by lemma~\ref{lemmaABBAValuesAgreement} if any unblocked node $\mathcal{P}_i\in\mathsf{UNL}_k^{\infty}$ broadcasts $AUX(v,r)$ then eventually every unblocked node $\mathcal{P}_j\in\mathsf{UNL}_k^{\infty}$ adds $v$ to $\mathsf{values}_j^r$. Thus for any unblocked node $\mathcal{P}_j\in\mathsf{UNL}_k^{\infty}$, every unblocked node in $\mathsf{UNL}_j$ will broadcast $AUX(v,r)$ for some $v$ which is eventually added to $\mathsf{values}_j^r$, so eventually $\mathcal{P}_j$ can progress to step~\ref{ABBAconfStep} since there are at least $q_S$ unblocked nodes in every essential subset $S\in\mathsf{ES}_j$.
\end{proof}

Finally, we make three quick lemmas about the confirmation phase.

\begin{lemma}\label{lemmaABBAConfAgreement}
If two honest nodes $\mathcal{P}_i$ and $\mathcal{P}_j$ are linked, then if $\mathcal{P}_i$ continues to step~\ref{ABBAfinal} in round $r$ with $\mathsf{values}_i^r=\{v\}$, $\mathcal{P}_j$ cannot continue to step~\ref{ABBAconfStep} in round $r$ with $\mathsf{values}_j^r=\{\neg v\}$.
\end{lemma}
\begin{proof}
Identical to the proof of lemma~\ref{lemmaABBAAuxAgreement}.
\end{proof}

\begin{lemma}\label{lemmaABBAConfTermination}
If $\mathcal{P}_k$ is strongly connected, every unblocked node in $\mathsf{UNL}_k^{\infty}$ gets to step~\ref{ABBAbeginLoop} for round $r$, and no unblocked nodes in $\mathsf{UNL}_k^{\infty}$ terminate in round $r$, then eventually every unblocked node in $\mathsf{UNL}_k^{\infty}$ either progresses to step~\ref{ABBAfinal} in round $r$ or terminates.
\end{lemma}
\begin{proof}
By an identical proof to lemma~\ref{lemmaABBAAuxTermination}, every unblocked node progresses to step~\ref{ABBArandom}. The lemma thus follows from CRS-Termination.
\end{proof}

The final lemma for this phase shows why the confirmation phase is needed. It prevents the adversary from ``gaming" the CRS to learn the value it returns in advance and using that information to artificially coordinate the system to prevent termination.

\begin{lemma}\label{lemmaABBAConfUnpredictability}
If $\mathcal{P}_k$ is strongly connected and some healthy node $\mathcal{P}_i\in\mathsf{UNL}_k^{\infty}$ progresses to step~\ref{ABBAfinal} in round $r$ with $\mathsf{values}_i^r=\{v\}$, then $|\mathrm{Pr}[s_r=v]-1/2|\leqslant\epsilon$ for some negligible $\epsilon$.
\end{lemma}
\begin{proof}
	In order for $\mathcal{P}_i$ to progress to step~\ref{ABBAfinal} in round $r$ with $\mathsf{values}_i^r=\{v\}$, $\mathcal{P}_i$ must have received strong support for $CONF(\{v\},r)$. By strong connectivity of $\mathcal{P}_k$, then any healthy node $\mathcal{P}_j\in\mathsf{UNL}_k^{\infty}$ that samples $\rho_r$ in step~\ref{ABBArandom} must have done so \textit{after} receiving strong support for $CONF(\{v\},r)$ from some healthy node in $\mathsf{UNL}_k^{\infty}$. By lemma~\ref{lemmaABBAAuxAgreement}, it cannot be the case that one healthy node in $\mathsf{UNL}_k^{\infty}$ broadcast $CONF(\{0\},r)$ while another healthy node broadcast $CONF(\{1\},r)$; thus the value of $v$ must have been determined \textit{before} $\mathcal{P}_j$ sampled $\rho_r$. Since $\rho_r$ samples randomly from $\{0,1\}$, by CRS-Randomness $|\textrm{Pr}[s_r=v]-1/2|\leqslant \epsilon$ for negligible $\epsilon$.
\end{proof}

We need two more quick lemmas that don't tie into either of the above ``phases", but rather deal with the correctness of the overall algorithm.

\begin{lemma}\label{lemmaABBAFinishValidity}
	If $\mathcal{P}_i$ is unblocked and outputs the value $v$, then there is some chain of unblocked nodes $\mathcal{P}_i=\mathcal{P}_{i_0},\mathcal{P}_{i_1},...,\mathcal{P}_{i_n}$, where for all $k\leqslant n$, $\mathcal{P}_{i_{k}}\in\mathsf{UNL}_{i_{k-1}}$, and the node $\mathcal{P}_{i_n}$ broadcast $FINISH(v)$ due to the logic in step~\ref{ABBAfinal}.
\end{lemma}
\begin{proof}
	Identical to the proof of lemma~\ref{lemmaABBAValuesValidity}.
\end{proof}

\begin{lemma}\label{lemmaABBAOutputConsistency}
	If $\mathcal{P}_k$ is strongly connected, and in some round $r$ a healthy node $\mathcal{P}_i\in\mathsf{UNL}_k^{\infty}$ gets to step~\ref{ABBAfinal} with $\mathsf{values}_i^r=\{s_r\}$ where $s_r$ is the value obtained from the random oracle $\rho_r$, then for every $r'>r$, any healthy node $\mathcal{P}_j\in\mathsf{UNL}_k^{\infty}$ that begins round $r'$ does so with $\mathsf{est}_j^{r'}=s_r$.
\end{lemma}
\begin{proof}
	Suppose in round $r'$ every healthy node $\mathcal{P}_j\in\mathsf{UNL}_k^{\infty}$ that begins round $r'$ does so with $\mathsf{est}_j^{r'}=s_r$. By taking the contrapositive of lemma~\ref{lemmaABBAValuesValidity}, one finds that every healthy node that gets to step~\ref{ABBAfinal} in round $r'$ must do so with $\mathsf{values}_{r'}=\{s_r\}$. Thus every healthy node $\mathcal{P}_j$ that begins round $r'+1$ does so with $\mathsf{est}_j^{r'+1}=s_r$.
	
	Therefore by induction it suffices to show that if in some round $r$ a healthy node $\mathcal{P}_i\in\mathsf{UNL}_k^{\infty}$ gets to step~\ref{ABBAfinal} with $\mathsf{values}_r=\{s_r\}$, then every healthy node $\mathcal{P}_j\in\mathsf{UNL}_k^{\infty}$ that begins round $r+1$ does so with $\mathsf{est}_j^{r+1}=s_r$. But by lemma~\ref{lemmaABBAAuxAgreement} and the assumption that $\mathcal{P}_k$ is strongly connected, any healthy node $\mathcal{P}_j\in\mathsf{UNL}_k^{\infty}$ that gets to step~\ref{ABBAfinal} in round $r$ must do so with either $\mathsf{values}_j^r=\{s_r\}$ or $\mathsf{values}_j^r=\{0,1\}$. In the former case, $\mathcal{P}_j$ continues to round $r+1$ with $\mathsf{est}_j^{r+1}=s_r$. In the latter case, $\mathcal{P}_j$ takes the value obtained from $\rho_r$ as $\mathsf{est}_j^{r+1}$; but by CRS-Agreement $\mathcal{P}_j$ outputs the same random value as $\mathcal{P}_i$, so again $\mathcal{P}_j$ continues to round $r+1$ with $\mathsf{est}_j^{r+1}=s_r$.
\end{proof}

Now with all of those lemmas out of the way, we can finally prove the correctness of the overall algorithm.

\begin{proposition}\label{propABBAStrongValidity}
	If $\mathcal{P}_i$ is unblocked and outputs $v$, then there is some chain of unblocked nodes $\mathcal{P}_i=\mathcal{P}_{i_0},\mathcal{P}_{i_1},...,\mathcal{P}_{i_n}$, where for all $k\leqslant n$, $\mathcal{P}_{i_{k}}\in\mathsf{UNL}_{i_{k-1}}$, and the node $\mathcal{P}_{i_n}$ input $v$.
\end{proposition}
\begin{proof}
	By lemma~\ref{lemmaABBAFinishValidity}, we can construct a chain of unblocked nodes $\mathcal{P}_i=\mathcal{P}_{i_0},\mathcal{P}_{i_1},...,\mathcal{P}_{i_{n_{r+1}}}$, where for all $k\leqslant n_{r+1}$, $\mathcal{P}_{i_{k}}\in\mathsf{UNL}_{i_{k-1}}$, and the node $\mathcal{P}_{i_n}$ broadcast $FINISH(v)$ due to the logic in step~\ref{ABBAfinal} in round $r$ for some $r\geqslant 0$. In particular,  $\mathcal{P}_{i_{n_{r+1}}}$ gets to step~\ref{ABBAfinal} in round $r$ with $\mathsf{values}_{i_{n_{r+1}}}^{r}=\{v\}$.
	
	We work backwards from $r$ to extend the chain until it reaches an unblocked node that input $v$.
	
	Let $r'\leqslant r$ and suppose $\mathcal{P}_{i_{n_{r'+1}}}$ is some unblocked node that gets to step~\ref{ABBAfinal} in round $r'$ with $v\in\mathsf{values}_{i_{n_{r'+1}}}^{r'}$. By lemma~\ref{lemmaABBAValuesValidity}, there is some chain of unblocked nodes $\mathcal{P}_{i_{n_{r'+1}}},\mathcal{P}_{i_{n_{r'+1}}+1},...,\mathcal{P}_{i_{n_{r'}}}$, where for all $k\leqslant n$, $\mathcal{P}_{i_{k}}\in\mathsf{UNL}_{i_{k-1}}$ and $\mathsf{est}_{i_{r'}}^{r'}=v$. But then either $r'=0$ and $\mathcal{P}_{i_{r'}}$ input $v$, or $r'>0$ and $\mathcal{P}_{i_{r'}}$ must have gotten to step~\ref{ABBAfinal} in round $r'-1$ with $v\in\mathsf{values}_{i_{r'}}^{r'-1}$.
	
	By repeating the above logic until we reach $r'=0$, we build out a chain $\mathcal{P}_i=\mathcal{P}_{i_0},\mathcal{P}_{i_1},...,\mathcal{P}_{i_{n_0}}$ satisfying the requirements of the proposition.
\end{proof}

\begin{proposition}\label{propABBATermination}
	If $\mathcal{P}_k$ is strongly connected and every unblocked node in $\mathsf{UNL}_k^{\infty}$ provides some input to the algorithm, then eventually every unblocked node in $\mathsf{UNL}_k^{\infty}$ terminates with probability $1$.
\end{proposition}
\begin{proof}
	Note that by lemma~\ref{lemmaABBAOutputConsistency}, any two unblocked nodes that broadcast $FINISH$ messages due to the logic in step~\ref{ABBAfinal} must broadcast the same $FINISH$ message. Thus by the same proof as proposition~\ref{propRBCReliability}, if any unblocked node in $\mathsf{UNL}_k^{\infty}$ terminates then all unblocked nodes in $\mathsf{UNL}_k^{\infty}$ terminate.
	
	Once every unblocked node in $\mathsf{UNL}_k^{\infty}$ provides some input in round $0$ then by applying lemma~\ref{lemmaABBAAuxTermination} inductively one sees that for every $r\geqslant 0$, either all nodes get to round $r$ or some unblocked node in $\mathsf{UNL}_k^{\infty}$ terminates before then. By the preceding paragraph, we derive that either every unblocked node in $\mathsf{UNL}_k^{\infty}$ eventually terminates, or every unblocked node in $\mathsf{UNL}_k^{\infty}$ gets to round $r$ for every $r\geqslant 0$.
	
	Suppose in some round $r$ every unblocked node $\mathcal{P}_j\in\mathsf{UNL}_k^{\infty}$ gets to step~\ref{ABBAfinal} with $\mathsf{values}_j^r=\{0,1\}$. Then every unblocked node in $\mathsf{UNL}_k^{\infty}$ will begin round $r+1$ with estimate set to the random oracle value from round $r$, so in particular every unblocked node begins round $r+1$ with a common value $s$ for their estimates. As in the proof of lemma~\ref{lemmaABBAOutputConsistency}, this implies that for all $r'>r$, every node will get to step~\ref{ABBAfinal} with $\mathsf{values}^{r'}=\{s\}$. Thus as soon as the CRS $\rho_{r'}$ returns $s$ for some $r'>r$---which happens within a finite number of rounds with probability $1$, and in fact takes only $2+\epsilon$ rounds in expectation for a negligible $\epsilon$---every unblocked node in $\mathsf{UNL}_k^{\infty}$ broadcasts $FINISH(s)$, allowing every unblocked node to terminate.
	
	Now on the other hand if in round $r$ there is some unblocked node $\mathcal{P}_i\in\mathsf{UNL}_k^{\infty}$ that gets to step~\ref{ABBAfinal} with $\mathsf{values}_i^r=\{v\}$, then either $s_r=v$, in which case by lemma~\ref{lemmaABBAOutputConsistency} and lemma~\ref{lemmaABBAValuesValidity} every unblocked node in $\mathsf{UNL}_k^{\infty}$ will get to step~\ref{ABBAfinal} in round $r'$ with $\mathsf{values}^{r'}=\{s\}$ for every $r'>r$, and as in the previous paragraph every unblocked node terminates with probability $1$. Otherwise the oracle in round $r$ returns $\neg v$, in which case the nodes go into the next round with some arbitrary state. However, by lemma~\ref{lemmaABBAConfUnpredictability} there is at least $1/2-\epsilon$ chance of the first option occurring, so with probability $1$ every unblocked node in $\mathsf{UNL}_k^{\infty}$ eventually terminates.
\end{proof}

\begin{theorem}\label{thmABBA}
	The protocol defined in section~\ref{sectionProtocol-ABBA} satisfies the properties of an asynchronous Byzantine binary agreement algorithm in the open network model.
\end{theorem}
\begin{proof}
	Agreement is proven in proposition~\ref{propABBAAgreement}. Termination is proven in proposition~\ref{propABBATermination}. Strong Validity (and hence plain Validity as well) is proven in proposition~\ref{propABBAStrongValidity}.
\end{proof}

\subsection{Democratic Atomic Broadcast}

\subsubsection{Definition}\label{sectionProtocol-DABCDef}

Although we loosely defined the DABC problem in section~\ref{sectionModel}, at the time we were unable to explicitly describe the network assumptions required for each property to hold, so we reiterate the problem definition and clarify the assumptions now.

As stated in section~\ref{sectionModel}, a protocol that solves DABC allows proposers to broadcast amendments to the network. Each node can choose to either support or oppose each amendment it receives, and then each node over time ratifies some of those amendments and assigns each ratified amendment an activation time, according to the following properties:
\begin{itemize}
	\item DABC-Agreement: If $\mathcal{P}_k$ is strongly connected and some healthy node in $\mathsf{UNL}_k^{\infty}$ ratifies an amendment $A$ an assigns it the activation time $\tau$, then eventually every unblocked node in $\mathsf{UNL}_k^{\infty}$ also ratifies $A$  with probability $1$ and assigns it the activation time $\tau$.
	\item DABC-Linearizability: If any honest node ratifies an amendment $A$ before ratifying some other amendment $A'$, then every other honest node linked to it ratifies $A$ before $A'$.
	\item DABC-Democracy: If any healthy node $\mathcal{P}_i$ is weakly connected and ratifies an amendment $A$, then there exists some essential subset $S\in \mathsf{ES}_i$ such that the majority of all honest nodes in $S$ supported $A$ being ratified, and further supported $A$ being ratified in the context of all the amendments ratified before $A$.
	\item DABC-Liveness: If $\mathcal{P}_k$ is strongly connected and every unblocked node in $\mathsf{UNL}_k^{\infty}$ supports some unratified amendment $A$, then eventually every unblocked node in $\mathsf{UNL}_k^{\infty}$ ratifies a new amendment with probability $1$.
	\item DABC-Full-Knowledge: For every time $\tau$, a healthy node that is weakly connected can wait some amount of time and afterwards know that it is aware of every amendment that will be ratified with an activation time less than $\tau$. Further, if $\mathcal{P}_k$ is strongly connected, then any unblocked node in $\mathsf{UNL}_k^{\infty}$ only needs to wait a finite amount of time with probability $1$.
\end{itemize}

To solve DABC, we use a reduction to DRBC and a different agreement protocol called \textbf{external validity multi-valued Byzantine agreement} or \textbf{MVBA}. A protocol that solves MVBA allows each node $\mathcal{P}_i$ to dynamically maintain a set $\mathsf{values}_i^0$ known as its \textbf{valid inputs}, and then come to consensus on some value that everyone in the network considers a valid input. We assume that these sets satisfy the following ``reliability" and ``validity" conditions:
\begin{itemize}
	\item Assumed-Reliability: If $\mathcal{P}_k$ is strongly connected and any healthy node $\mathcal{P}_i\in\mathsf{UNL}_k^{\infty}$ adds $A$ to $\mathsf{values}_i^0$, then eventually every unblocked node $\mathcal{P}_j\in\mathsf{UNL}_k^{\infty}$ adds $A$ to $\mathsf{values}_j^0$.
	\item Assumed-Validity: If $\mathcal{P}_k$ is strongly connected and any unblocked node $\mathcal{P}_i\in\mathsf{UNL}_k^{\infty}$ adds $A$ to $\mathsf{values}_i^0$, then there is some unblocked node $\mathcal{P}_j\in\mathsf{UNL}_k^{\infty}$ such that for every $S\in\mathsf{ES}_j$, the majority of unblocked nodes in $S$ ``suggested" $A$ before beginning the protocol.
\end{itemize}

Assumed-Reliability is important for ensuring eventual termination. Assumed-Validity is only actually needed in appendix~\ref{sectionAppendix2} where we use it for proving a result about the relative efficiency of our MVBA algorithm.

Formally, under the above assumptions, an MVBA protocol is a protocol that allows nodes to output some value according to the following properties:
\begin{itemize}
	\item MVBA-Consistency: No two honest, linked nodes can output different values.
	\item MVBA-Termination: If $\mathcal{P}_k$ is strongly connected, $\mathsf{values}_i^0$ has bounded size for every unblocked node $\mathcal{P}_i\in\mathsf{UNL}_k^{\infty}$, and eventually some value $A$ is in $\mathsf{values}_i^0$ for every unblocked node $\mathcal{P}_i\in\mathsf{UNL}_k^{\infty}$; then eventually every unblocked node in $\mathsf{UNL}_k^{\infty}$ terminates with probability $1$.
	\item MVBA-Validity: If $\mathcal{P}_i$ outputs $A$, then $A\in\mathsf{values}_i^0$.
\end{itemize}

Note that our definition of MVBA is fairly different from that of Cachin et al. \citep{Cachin2001}. Cachin et al. don't assume any sort of reliability for their valid input sets, and instead use cryptographic proofs to guarantee that any honest node's input can be verified as valid by everyone else. Our different definition is necessitated by the lack of sufficiently expressive cryptographic proofs in our domain. In the complete network model, a protocol that satisfies our definition can trivially be applied in place of a protocol satisfying Cachin et al.'s definition, simply by specifying an honest node $\mathcal{P}_i$ adds a value $A$ to $\mathsf{values}_i^0$ if it receives a valid proof for $A$'s validity. This might not satisfy Assumed-Validity, but since Assumed-Validity is only needed for efficiency this is not a huge issue, and for most use cases of MVBA it \textit{will} satisfy Assumed-Validity.

The idea behind the reduction of DABC to MVBA is that each proposer uses DRBC to broadcast their amendment $A$ along with a \textbf{slot number} $n_A$ that identifies where in the total ordering of amendments $A$ is intended to be ratified. Then for each slot number $n$, a node waits until it has ratified an amendment with every earlier slot number and then supports $A$ if and only if it supports $A$ in the context of the amendments ratified before slot $n_A$. The nodes begin an MVBA instance tagged with $n_A$, and $\mathcal{P}_i$ sets $\mathsf{values}_i^0$ to be the set of all the amendments with slot number $n$ that $\mathcal{P}_i$ accepts through DRBC, and ratifies whichever amendment is eventually output from MVBA. Assumed-Reliability for the valid inputs holds immediately by RBC-Reliability, and Assumed-Validity holds if ``suggesting" refers to the act of supporting in DRBC. The actual reduction requires a slight extension to guarantee Full-Knowledge. The full reduction is described formally in section~\ref{sectionProtocol-FK}.

An alternative to specifying the slot number would be to include the hash of the most recently ratified amendment proposal in each amendment proposal. This would satisfy all the same properties, but may be more intuitive coming from the ``blockchain cannon". It also could make it easier to tell when the system has broken (since nodes that disagree with you will have different hashes for the previous amendment) which could help nodes to panic and halt everything until the system can be fixed rather than simply charging ahead and possibly increasing the amount of damage that needs to be repaired. We use the slot-based definition in this paper since it's notationally simpler, and leave the choice of which definition to actually use up to the implementors.

\subsubsection{Multi-Valued Agreement}\label{sectionProtocol-MIBA}

We now present our protocol for solving MVBA. To the author's knowledge, this protocol is not derived from any other complete network protocol. It relies upon a reduction of MVBA to ABBA.

Similar to the ABBA protocol from section~\ref{sectionProtocol-ABBA}, the MVBA protocol proceeds in rounds. The protocol uses a sequence of CRS instances to give a ``random index" to the values for each round. Specifically, we assume the existence of a collision resistant hash function $H$ in the random oracle model \citep{Bellare1993}. In other words, for every $x$, $H(x)$ is modeled as a true random variable drawn uniformly from the image of $H$, which can \textit{only} be derived by explicitly asking an imagined oracle to apply $H$ to a chosen input $x$. Let $\mathcal{S}$ be a uniform probability space over a set of size which is super-polynomial in the security parameter. For every $r\geqslant 0$, let $\rho_r$ be a CRS defined over $\mathcal{S}$. Then if $s_r$ is the value received from $\rho_r$, we define the functions $\mathcal{I}_r$ by $\mathcal{I}_r(A)=H(A||s_r)$. By the assumption that $\mathcal{S}$ is uniform over a super-polynomial set and the CRS-Randomness property, the adversary can only produce $s_r$ in advance with negligible probability. Thus for any $A$ the adversary can only produce $A||s_r$ with negligible probability, so until some healthy node samples $\rho$, with overwhelming probability $\mathcal{I}_r(A)$ is a sequence of independent, uniformly sampled random variables for every $r\geqslant 0$.

It is worth noting that unlike the ABBA protocol, the randomness of the CRS $\rho_r$ is not needed to guarantee termination. As long as $H$ is collision resistant, then even if the random values are known in advance there is no way for the network adversary to make the protocol continue for an infinite number of rounds. However, without the randomness of $\rho_r$, termination can take a number of rounds linear in the number of valid inputs, whereas with the randomness assumption termination only takes at most an expected logarithmic number of rounds. We prove this in section~\ref{sectionAppendix2}.

To run MVBA, the node $\mathcal{P}_i$ runs the following protocol.
\begin{enumerate}
	\item Set $\mathsf{values}_i^r=\emptyset$ for all $r>0$, and set $r=0$.
	\item\label{MVBAbeginLoop} Wait until $\mathsf{values}_i^r$ contains some value $A$, then broadcast $ELECT(A,r)$ if we have not yet broadcast $ELECT(\mathunderscore,r)$.
	\item\label{MVBAelecting} For every essential subset $S$, wait until there exists some subset $T\subseteq S$, such that $\vert T\vert\geqslant q_S$, we received $ELECT(\mathunderscore,r)$ from every node in $T$, and if any node in $T$ sent us $ELECT(A',r)$ for some $A'$, then $A'\in\mathsf{values}_i^r$.
	
	After waiting, if $\mathsf{values}_i^r=\{A\}$ for some value $A$, broadcast $FINISH(A,r)$. Otherwise broadcast $CONT(\mathsf{values}_i^r,r)$.
	\item\label{MVBAvoting} Upon receiving strong support for $FINISH(A,r)$, vote $1$ in an ABBA instance tagged with $(``STOP",r)$. Otherwise, upon receiving $CONT(C,r)$ from any node where $|C|\geqslant 2$ and $C\subseteq\mathsf{values}_i^r$, broadcast $CONT(\mathsf{values}_i^r,r)$ and then vote $0$ in the ABBA instance tagged with $(``STOP",r)$.
	\item\label{MVBALongStep} Wait until the ABBA instance tagged with $(``STOP",r)$ terminates. If it terminates on $1$, wait until we receive weak support for $FINISH(A,r)$ for some value $A$, then broadcast $FINISH(A,r)$ if we haven't already broadcast $FINISH(\mathunderscore,r)$; then wait until we receive strong support for $FINISH(A,r)$ where $A\in\mathsf{values}_i^0$, and then finally output $A$ and terminate.
	
	Otherwise if the ABBA instance terminates on $0$, wait until we receive $CONT(C,r)$ from some node, where $|C|\geqslant 2$ and $C\subseteq\mathsf{values}_i^r$. Then broadcast $CONT(\mathsf{values}_i^r,r)$; further, if $\mathsf{values}_i^r$ later grows then each time broadcast $CONT(\mathsf{values}_i^r,r)$ with the updated set. For every essential subset $S$, wait until there exists some set $C\subseteq \mathsf{values}_i^r$ such that we've received strong support for $CONT(C,r)$, then query the random oracle $\rho_r$ for $s_r$, set $\mathsf{est}_i^{r+1}$ to the value in $\mathsf{values}_i^r$ with minimum $\mathcal{I}_r$ index, and broadcast $INIT(\mathsf{est}_i^{r+1},r+1)$.
	\item\label{MVBAinit1} Upon receiving weak support for $INIT(A,r+1)$ for an arbitrary value $A$, or upon adding $A$ to $\mathsf{values}_i^r$ for some value $A$ such that $\mathcal{I}_r(A)<\mathcal{I}_r(\mathsf{est}_i^r)$, broadcast $INIT(A,r+1)$ if we have not already done so.
	\item\label{MVBAinit2} Upon receiving strong support for $INIT(A,r+1)$, add $A$ to $\mathsf{values}_i^{r+1}$, set $r=r+1$, and return to step~\ref{MVBAbeginLoop} if we have not yet done so in this round.
\end{enumerate}

The above protocol is again defined asynchronously, so that once you get to some step in the protocol you keep running that step forever. This is important since for example you might need to add more values to $\mathsf{values}_i^r$ than simply the first one that you add before jumping back to step~\ref{MVBAbeginLoop}.

One easy optimization is to begin broadcasting messages for round $r+1$ without waiting for the round $r$ ABBA to terminate. As long as we follow the termination procedure for the first round in which ABBA terminates on $1$, this can cut down the latency by a significant fraction without affecting the correctness of the protocol.

\subsubsection{Analysis}\label{sectionProofs-MIBA}

We will first prove the correctness of the MVBA algorithm, and then at the end we will prove the correctness of our reduction from DABC to MVBA.

The following proposition shows that consistency is a local property. Thus, although forward progress may depend on the configuration of other nodes in the network, a node can at least guarantee that the amendments it observes are consistent with the rest of the network as long as \textit{it alone} is well configured.

\begin{proposition}\label{propDABCConsistency}
	If two honest nodes $\mathcal{P}_i,\mathcal{P}_j$ are linked, then if $\mathcal{P}_i$ outputs $A$, $\mathcal{P}_j$ cannot output any $A'\neq A$.
\end{proposition}
\begin{proof}
	Suppose $\mathcal{P}_i$ outputs $A$. Then there must be some round $r$ where $\mathcal{P}_i$ saw that ABBA instance tagged with $(``STOP",r)$ terminate with $1$, the ABBA instances tagged with $(``STOP",r')$ for every $r'<r$ terminate with $0$, and $\mathcal{P}_i$ received strong support for $FINISH(A,r)$. By proposition~\ref{propABBAAgreement}, $\mathcal{P}_j$ cannot see different ABBA outputs, so if $\mathcal{P}_j$ outputs $A'$ it must do so due to receiving strong support for $FINISH(A',r)$. Since honest nodes can only broadcast a single $FINISH(\mathunderscore,r)$ message, by lemma~\ref{lemmaSupportBlocking}, $A'=A$.
\end{proof}

We now develop a few lemmas before we can prove the stronger consensus-properties of MVBA.

\begin{lemma}\label{lemmaDABCValidAgreement}
	If $\mathcal{P}_k$ is strongly connected and any healthy node $\mathcal{P}_i\in\mathsf{UNL}_k^{\infty}$ adds $A$ to $\mathsf{values}_i^r$ for some $r\geqslant 0$, then every unblocked node $\mathcal{P}_j\in\mathsf{UNL}_k^{\infty}$ will eventually add $A$ to $\mathsf{values}_j^r$.
\end{lemma}
\begin{proof}
	For $r=0$ this follows by Assumed-Reliability. For $r>0$, the proof is identical to the proof of lemma~\ref{lemmaABBAValuesAgreement}.
\end{proof}

For each $r\geqslant 0$ and each node $\mathcal{P}_i$, let $S_i^r$ be the set of all values that are eventually added to $\mathsf{values}_i^r$.

\begin{lemma}\label{lemmaDABCFiniteness}
	For any strongly connected node $\mathcal{P}_k$, if $S_i^0$ is finite for every unblocked node $\mathcal{P}_i\in\mathsf{UNL}_k^{\infty}$, then for every $r\geqslant 0$ and every unblocked node $\mathcal{P}_j\in\mathsf{UNL}_k^{\infty}$, $|S_j^r|>|S_j^{r+1}|$.
\end{lemma}
\begin{proof}	
	Since $\mathcal{P}_k$ is strongly connected, for every healthy node $\mathcal{P}_i\in\mathsf{UNL}_k^{\infty}$ and every unblocked node $\mathcal{P}_j\in\mathsf{UNL}_k^{\infty}$, $S_i^r\subseteq S_j^r$ by lemma~\ref{lemmaDABCValidAgreement}. Thus if a value $A$ is not in $S_j^r$, then no healthy node in $\mathsf{UNL}_k^{\infty}$ will ever broadcast $INIT(A,r+1)$, so $\mathcal{P}_j$ will never add $A$ to $\mathsf{values}_j^{r+1}$ implying $A\notin S_j^{r+1}$. Thus $S_j^{r+1}\subseteq S_j^r$, so to show that $|S_j^r|>|S_j^{r+1}|$ it suffices to show that there is some value in $S_j^r$ that is not in $S_j^{r+1}$.
	
	For a given $r\geqslant 0$, let $A_{max}$ be the value with maximum $\mathcal{I}_r$ index in $S_j^r$. By step~\ref{MVBALongStep} of the protocol, an honest node $\mathcal{P}_i\in\mathsf{UNL}_k^{\infty}$ only sets $\mathsf{est}_i^{r+1}$ to some value $A$ if $|\mathsf{values}_i^r|\geqslant 2$ and $A$ is the value with minimum $\mathcal{I}_r$ index in $\mathsf{values}_i^r$. But $S_j^r\supseteq S_i^r\supseteq \mathsf{values}_i^r$, so if $|\mathsf{values}_i^r|\geqslant 2$ then the value with minimum $\mathcal{I}_r$ index in $\mathsf{values}_i^r$ must have index strictly less than $A_{max}$ (strictness comes from collision resistance of $H$). Thus no honest node in $\mathsf{UNL}_k^{\infty}$ will ever broadcast $INIT(A_{max},r+1)$, so $\mathcal{P}_j$ can never add $A_{max}$ to $\mathsf{values}_j^{r+1}$, so $A_{max}\notin S_j^{r+1}$.
\end{proof}

\begin{lemma}\label{lemmaDABCElectTermination}
	If $\mathcal{P}_k$ is strongly connected and every unblocked node in $\mathsf{UNL}_k^{\infty}$ gets to step~\ref{MVBAelecting} in round $r\geqslant 0$, then eventually either every unblocked node in $\mathsf{UNL}_k^{\infty}$ terminates in round $r$ or every unblocked node in $\mathsf{UNL}_k^{\infty}$ progresses to round $r+1$, with probability $1$.
\end{lemma}
\begin{proof}
	By assumption eventually every unblocked node in $\mathsf{UNL}_k^{\infty}$ broadcasts $ELECT(\mathunderscore,r)$. Further, by lemma~\ref{lemmaDABCValidAgreement}, if any unblocked node in $\mathsf{UNL}_k^{\infty}$ broadcasts $ELECT(A,r)$ then eventually every unblocked node $\mathcal{P}_j\in\mathsf{UNL}_k^{\infty}$ adds $A$ to $\mathsf{values}_j^r$. Thus for any unblocked node $\mathcal{P}_j\in\mathsf{UNL}_k^{\infty}$, every unblocked node in $\mathsf{UNL}_j$ will eventually broadcast $ELECT(A,r)$ for some $A$ which is eventually in $\mathsf{values}_j^r$, allowing $\mathcal{P}_j$ to progress to step~\ref{MVBAvoting}.
	
	Since a healthy node only broadcasts $FINISH(A,r)$ for some value $A$ if some healthy node in its UNL broadcast $FINISH(A,r)$ first or it received strong support for $ELECT(A,r)$, by the same proof as in lemma~\ref{propRBCReadyBlocking} every healthy node in $\mathsf{UNL}_k^{\infty}$ that broadcasts a $FINISH(A,r)$ message does so for a common value $A$.
	
	Since every unblocked node $\mathsf{P}_i\in\mathsf{UNL}_k^{\infty}$ gets to step~\ref{MVBAvoting} in round $r$ by the first paragraph, every unblocked node in $\mathsf{UNL}_k^{\infty}$ either broadcasts $FINISH(A,r)$ for some common value $A$ or $CONT(\mathsf{values}_r,r)$ where $|\mathsf{values}_i^r|\geqslant 2$. For a given unblocked node $\mathcal{P}_j\in\mathsf{UNL}_k^{\infty}$, if every unblocked node in $\mathsf{UNL}_j$ broadcasts $FINISH(A,r)$ then $\mathcal{P}_j$ eventually receives strong support for $FINISH(A,r)$ and votes $1$ in the ABBA instance tagged with $(``STOP",r)$. Otherwise $\mathcal{P}_j$ eventually receives some $CONT(C,n,r)$ from some unblocked node $\mathcal{P}_i\in\mathsf{UNL}_j$. Since $\mathcal{P}_i$ is healthy, $C$ must have been a subset of $\mathsf{values}_i^r$, so by lemma~\ref{lemmaDABCValidAgreement} eventually $C\subseteq\mathsf{values}_j^r$, so $\mathcal{P}_j$ eventually sees the $CONT$ message as valid and votes $0$ in the ABBA instance tagged with $(``STOP",r)$. Thus every unblocked node in $\mathsf{UNL}_k^{\infty}$ eventually votes in the ABBA instance, and by proposition~\ref{propABBATermination}, the instance eventually terminates with probability $1$.
	
	Suppose the ABBA instance terminates on $1$. Then by proposition~\ref{propABBAStrongValidity}, there must have been some unblocked node in $\mathsf{UNL}_k^{\infty}$ that voted $1$ and thus received strong support for $FINISH(A,r)$. But if any unblocked node in $\mathsf{UNL}_k^{\infty}$ receives strong support for $FINISH(A,r)$ then by a similar proof as in proposition~\ref{propRBCReliability}, eventually every other unblocked node in $\mathsf{UNL}_k^{\infty}$ will receive strong support for $FINISH(A,r)$. Since an honest node $\mathcal{P}_i$ only broadcasts $FINISH(A,r)$ if $A\in\mathsf{values}_i^r\subseteq\mathsf{values}_i^0$, by lemma~\ref{lemmaDABCValidAgreement} eventually every unblocked node in $\mathsf{UNL}_k^{\infty}$ adds $A$ as a valid input. Thus after seeing that the ABBA instance tagged $(``STOP",r)$ terminated on $1$, eventually every unblocked node in $\mathsf{UNL}_k^{\infty}$ outputs $A$ in round $r$ and terminates.
	
	If on the other hand the ABBA instance terminates on $0$, then by proposition~\ref{propABBAStrongValidity}, for every unblocked node $\mathcal{P}_i\in\mathsf{UNL}_k^{\infty}$ there is a chain of unblocked nodes $\mathcal{P}_i=\mathcal{P}_{i_0},...,\mathcal{P}_{i_n}$ where $\mathcal{P}_{i_k}\in\mathsf{UNL}_{i_{k-1}}$ for all $k\leqslant n$ and $\mathcal{P}_{i_n}$ voted $0$. But a healthy and correct node $\mathcal{P}_{i_n}$ only votes $0$ in the $(``STOP",r)$ ABBA instance if it has broadcast a $CONT$ message which by lemma~\ref{lemmaDABCValidAgreement}, eventually every unblocked node in $\mathsf{UNL}_k^{\infty}$ can recognize as valid. Thus this $CONT$ message can be passed back along the chain until it reaches $\mathcal{P}_i$, who eventually sees it as valid. By lemma~\ref{lemmaDABCValidAgreement}, eventually there is some set $S$ such that for every unblocked node $\mathcal{P}_j\in\mathsf{UNL}_k^{\infty}$, $\mathsf{values}_j^r=S$, so eventually $\mathcal{P}_i$ will receive strong support for $CONT(S,r)$ and proceed to step~\ref{MVBAinit1}.
	
	Let $\mathcal{P}_j\in\mathsf{UNL}_k^{\infty}$ be unblocked and let $A_{min}$ be the value with minimum $\mathcal{I}_r$ index in $S_j^r$. For every unblocked node $\mathcal{P}_i\in\mathsf{UNL}_k^{\infty}$, since $\mathcal{P}_i$ sets $\mathsf{est}_i^{r+1}$ by hypothesis and $S_j^r=S_i^r$ by lemma~\ref{lemmaDABCValidAgreement}, we have $\mathcal{I}_r(A_{min})\leqslant \mathcal{I}_r(\mathsf{est}_i^{r+1})$ so $\mathcal{P}_i$ eventually adds $A_{min}$ to $\mathsf{values}_i^r$. Thus eventually every unblocked node in $\mathsf{UNL}_k^{\infty}$ broadcasts $INIT(A_{min},r+1)$, so eventually $\mathcal{P}_j$ can add $A_{min}$ to $\mathsf{values}_j^{r+1}$ and progress to round $r+1$.
\end{proof}

\begin{proposition}\label{lemmaDABCTermination}
	If $\mathcal{P}_k$ is strongly connected and for every unblocked node $\mathsf{P}_i\in\mathsf{UNL}_k^{\infty}$ $\mathsf{values}_i^0$ is bounded in size and eventually nonempty, then eventually every unblocked node in $\mathsf{UNL}_k^{\infty}$ outputs some value with probability $1$.
\end{proposition}
\begin{proof}
	By lemma~\ref{lemmaDABCElectTermination}, either every unblocked node in $\mathsf{UNL}_k^{\infty}$ terminates in some round $r$ or for every $r\geqslant 0$ every unblocked node in $\mathsf{UNL}_k^{\infty}$ eventually gets to round $r$ with probability $1$.
	
	Therefore, by lemma~\ref{lemmaDABCFiniteness} and our assumption that $\mathsf{values}_i^0$ is bounded (i.e., $S_i^0$ is finite), eventually either every unblocked node in $\mathsf{UNL}_k^{\infty}$ terminates or every unblocked node in $\mathsf{UNL}_k^{\infty}$ gets to some round $r$ where $|S_k^r|\leqslant 1$ with probability $1$. If $|S_k^r|<1$, then no honest node can ever progress past step~\ref{MVBAbeginLoop}, implying that every unblocked node terminates (since otherwise there would be an $r\geqslant 0$ such that no unblocked node in $\mathsf{UNL}_k^{\infty}$ eventually gets to round $r$ with probability $1$).
	
	Thus, every unblocked node in $\mathsf{UNL}_k^{\infty}$ gets to some round $r$ where $|S_k^r|=1$ with probability $1$. Letting $S_k^r=\{A\}$, every unblocked node is guaranteed to broadcast $ELECT(A,r)$, so every unblocked node broadcasts $FINISH(A,r)$, so every unblocked node votes $1$ in the ABBA instance tagged with $(``STOP",r)$, and finally every unblocked node terminates in round $r$ and ratifies $A$.
\end{proof}

\begin{theorem}\label{thmMIBA}
The protocol defined in section~\ref{sectionProtocol-MIBA} satisfies the properties of an external validity multi-valued Byzantine agreement algorithm in the open network model.
\end{theorem}
\begin{proof}
Consistency is proven in proposition~\ref{propDABCConsistency}. Termination is proven in proposition~\ref{lemmaDABCTermination}.

Validity follows trivially from the fact that in step~\ref{MVBALongStep} we only accept the value $A$ if it is included in $\mathsf{values}_i^0$.
\end{proof}

\subsection{Reducing DABC to MVBA}
\subsubsection{Protocol}\label{sectionProtocol-FK}

Having developed our MVBA protocol, all that remains is to formalize our reduction of DABC to MVBA and prove its correctness. We begin first though with an intuitive discussion that helps to better understand our choice for how we guarantee Full-Knowledge for Cobalt.

As stated in section~\ref{sectionProtocol-DABCDef}, the basic idea of our reduction is to have the proposers distribute their amendment proposals using DRBC, and then use MVBA to agree on a single amendment for each slot. An obvious first option for agreeing on the activation time for an amendment $A$ is to include the activation time as part of the proposal for $A$. This easily guarantees agreement on activation times by the Agreement property of DRBC.

Unfortunately, there's no way to make such a system satisfy both Liveness and Full-Knowledge. For Full-Knowledge, nodes need to agree at some point after time $\tau$ on which amendments might be ratified with activation times earlier than $\tau$. If the proposal for $A$ specifies that $A$ must have activation time $\tau$, then the network adversary can thus just wait until the honest nodes have decided on which amendment could be ratified before time $\tau$, and then deliver $A$ to the honest nodes only after that point. Since no honest nodes knew about $A$ in time, there is then no way for $A$ to be ratified. Thus Liveness can't be guaranteed, since every amendment can be withheld long enough to cancel its validity.

Because of this problem, rather than requiring amendments to come packaged with an activation time, it becomes necessary to be able to agree cooperatively on an activation time for $A$ \textit{after} $A$ is received by the network. We now formally describe how we do this.

First, we assume there is some implementation-defined parameter $\tau_{int}$ that defines some interval duration. Making this parameter longer reduces contention going into consensus (which can speed up termination) and decreases network congestion, but making it too long can mean that you force you to wait longer before accepting  (which can slow down termination). Thus finding a good balance is important for optimal performance. In practice, setting $\tau_{int}$ to around $15$ seconds should give better performance than would be needed for any reasonable level of required urgency, while avoiding an unreasonable level of added network congestion.

We consider for every natural number $n$, there is a unique instance of MVBA that is designated for slot $n$. A proposer that wants to propose the amendment $A$ for slot $n_A$ runs DRBC to broadcast the message $(A,n_A)$. A node $\mathcal{P}_i$ supports this message in DRBC only if $\mathcal{P}_i$ has ratified an amendment for every slot below $n_A$, and $\mathcal{P}_i$ supports $A$ in the context of all of these previously ratified amendments. 

Let $P$ be a set that starts out empty. Upon accepting DRBC for $(A,n_A)$, $\mathcal{P}_i$ adds $(A,n_A)$ to $P$. For every time $\tau$ which is a multiple of $\tau_{int}$, upon arriving at time $\tau$, $\mathcal{P}_i$ runs the following protocol:
\begin{enumerate}
	\item Broadcast $CHECK(P,\tau)$.
	\item For a given pair $(A,n_A)$, once we have received a $CHECK(\mathunderscore,\tau)$ message that includes $(A,n_A)$ in its $P$ set from $q_S$ nodes in every essential subset $S\in\mathsf{ES}_i$, broadcast $ACCEPT(A,n_A,\tau)$. We may broadcast multiple $ACCEPT$ messages if the condition is also satisfied at some point for a different pair.
	\item Upon receiving weak support for $ACCEPT(A,n_A,\tau)$, broadcast $ACCEPT(A,n_A,\tau)$.
	\item Upon receiving strong support for $ACCEPT(A,n_A,\tau)$, add $(A,\tau)$ to $\mathsf{valid}_i^0$ in the MVBA instance for slot $n_A$, and remove any pairs from $P$ with slot $n_A$ (and don't add any new pairs to $P$ in the future that have slot $n_A$).
\end{enumerate}

We call the combination of the DRBC instances with the above protocol the \textbf{stamping protocol}. Effectively the stamping protocol just makes us continually try to pick out activation times for any supported amendment until eventually we see enough $ACCEPT$ messages that agree on the same timestamp so that we can use it for MVBA. Note that it is entirely possible with the above protocol to have multiple valid inputs that pertain to the same amendment and only differ in activation times. MVBA will choose a single activation time that everyone agrees upon, so this does not cause any issues.

Now to check which amendments are ratified by time $\tau$, we use a one-message \textbf{waiting protocol}: wait until, for every time $\tau'\leqslant \tau$ which is a multiple of $\tau_{int}$ and for every essential subset $S$, there exists some subset $T_{\tau'}\subseteq S$, such that $\vert T_{\tau'}\vert\geqslant q_S$, and from each node in $T_{\tau'}$ we received some message $CHECK(P,\tau')$ (possibly with different sets $P$ from different nodes) such that for every pair $(A,n_A)\in P$ we've ratified some amendment for the slot $n_A$.

Roughly speaking, the rationality behind these protocols is that if any healthy node broadcasts a $CHECK$ message for some amendment $A$, then we guarantee that some pair $(A,n_A,\tau)$ will eventually be accepted in the stamping protocol by all unblocked nodes. Therefore every unblocked node eventually provides some input into MVBA for the slot $n_A$, after which MVBA is guaranteed to terminate in a finite amount of time with probability $1$. This prevents infinite waiting in the waiting protocol. On the other hand, if any healthy node progresses past the waiting protocol for time $\tau$ without having seen some amendment $A$, then we guarantee that there could not have been enough $CHECK(\mathunderscore,\tau)$ messages containing $(A,n_A)$ for any healthy node to broadcast $ACCEPT(A,n_A,\tau)$, so $A$ cannot be accepted with timestamp $\tau$ by any healthy node.

Note that the above protocol usually requires waiting a short amount of time past $\tau$ for DABC to resolve before a node can learn all the amendments ratified before time $\tau$. A slight optimization would be to specify another duration parameter $\tau_{adv}$, and modify the protocol slightly so that an amendment that is accepted as $(A,n_A,\tau)$ actually has activation time $\tau+\tau_{adv}$, and the waiting protocol for time $\tau$ only waits for $\tau'\leqslant \tau-\tau_{adv}$. If $\tau_{adv}$ is set to the expected maximum amount of time that DABC should take to ratify some slot after all nodes provide input for that slot, then under normal conditions the waiting protocol for time $\tau$ will already be finished by time $\tau$.

\subsubsection{Analysis}\label{sectionProofs-FK}

We now prove the correctness of the full DABC protocol.

\begin{proposition}\label{propFKDRBCReliability}
	Outputs from the stamping protocol satisfy Assumed-Reliability and Assumed-Validity, if suggesting $(A,\tau)$ is defined to be broadcasting $CHECK(P,\tau)$ with $(A,n_A)\in P$.
\end{proposition}
\begin{proof}
	The mechanics of the $ACCEPT$ message in modified DRBC are identical to the mechanics of the $READY$ message in RBC, so the proof of Assumed-Reliability is the same as proposition~\ref{propRBCReliability}.
	
	For Assumed-Validity, suppose $\mathcal{P}_k$ is strongly connected and an unblocked node $\mathcal{P}_i\in\mathsf{UNL}_k^{\infty}$ adds $(A,\tau)$ to $\mathsf{values}_i^0$. Then some unblocked node $\mathcal{P}_j\in\mathsf{UNL}_k^{\infty}$ must have broadcast $ACCEPT(A,n_A,\tau)$, which it can only do having received messages suggesting $(A,\tau)$ from $q_S$ nodes in every essential subset $S\in\mathsf{ES}_j$. But for any node that broadcasts $CHECK(P,\tau)$ \textit{after} beginning MVBA for slot $n_A$, the stamping protocol necessitates that no pair in $P$ can have slot $n_A$. Thus $q_S$ nodes in every essential subset $S\in\mathsf{ES}_j$ suggested $(A,\tau)$ before beginning MVBA for slot $n_A$, from which Assumed-Validity follows from equation~\ref{eqTandQ1}.
\end{proof}

The following two lemmas are key to how the modified algorithm satisfies the Full Knowledge property.

\begin{lemma}\label{lemmaFKDRBCAgreement}
	If $\mathcal{P}_k$ is strongly connected and some healthy node in $\mathsf{UNL}_k^{\infty}$ broadcasts $CHECK(P,t)$, then for every $A\in P$, eventually every unblocked node in $\mathsf{UNL}_k^{\infty}$ accepts modified DRBC for some pair $(A,\mathunderscore)$.
\end{lemma}
\begin{proof}
	By proposition~\ref{propFKDRBCReliability}, if any healthy node in $\mathsf{UNL}_k^{\infty}$ accepts modified DRBC for some pair $(A,t')$ then eventually every unblocked node in $\mathsf{UNL}_k^{\infty}$ accepts modified DRBC for $(A,t')$. Thus it suffices to show that if some healthy node in $\mathsf{UNL}_k^{\infty}$ broadcasts $CHECK(P,t)$, then for every $A\in P$, eventually some healthy node in $\mathsf{UNL}_k^{\infty}$ accepts modified DRBC for some pair $(A,t')$ with $t'\geqslant t$.
	
	Note that if $\mathcal{P}_i$ is healthy and has not yet accepted some pair $(A,\mathunderscore)$, then $\mathcal{P}_i$ broadcasts $CHECK(P,t)$ if and only if it would have accepted unmodified DRBC for every $A\in P$ before time $t$. By proposition~\ref{propRBCReliability}, if $\mathcal{P}_i$ broadcasts $CHECK(P,t)$ then for every $A\in P$ either some unblocked node in $\mathsf{UNL}_k^{\infty}$ accepts some pair $(A,\mathunderscore)$ or eventually there is some $t'$ for which every unblocked node in $\mathsf{UNL}_k^{\infty}$ broadcasts some $CHECK(P,t')$ with $A\in P$. Thus every unblocked node in $\mathsf{UNL}_k^{\infty}$ broadcasts $ACCEPT(A,t')$, so eventually every unblocked node in $\mathsf{UNL}_k^{\infty}$ accepts $(A,t')$.
\end{proof}

\begin{lemma}\label{lemmaFKSufficientWaiting}
If $\mathcal{P}_k$ is healthy and weakly connected and receives strong support for $CHECK(\mathunderscore,t)$, then for any amendment $A$ that is not present in \textit{any} of the received $CHECK(\mathunderscore, t)$ messages, no $\mathcal{P}_k$ will never ever accept modified DRBC for $(A,t)$.
\end{lemma}
\begin{proof}
The proof of this is more or less the same as the proof of lemma~\ref{lemmaSupportBlocking}.

Suppose a healthy node $\mathcal{P}_i\in\mathsf{UNL}_k^{\infty}$ accepts modified DRBC for $(A,t)$. Then there must have been $q_S$ nodes in every essential subset $S\in\mathsf{ES}_i$ which broadcast some $CHECK(\mathunderscore,t)$ message including $A$. Since $\mathcal{P}_k$ is weakly connected, it is in particular fully linked to $\mathcal{P}_i$, so there is some $S\in\mathsf{ES}_k$ in which at least $q_S-t_S\geqslant t_S+1$ correct nodes broadcast some $CHECK(\mathunderscore,t)$ message including $A$ and $q_S\geqslant n_S-t_S$. Since honest nodes only broadcast a single $CHECK$ message for each timestamp, $\mathcal{P}_k$ thus can receive at most $n_S-(t_S+1)<q_S$ $CHECK(\mathunderscore,t)$ messages from nodes in $S$ that do not include $A$.
\end{proof}

\begin{lemma}\label{propDABCTermination}
If $\mathcal{P}_k$ is strongly connected and any healthy node in $\mathsf{UNL}_k^{\infty}$ broadcasts $CHECK(P,\tau)$ with some pair $(A,n_A)\in P$, then eventually every unblocked node in $\mathsf{UNL}_k^{\infty}$ ratifies some pair $(A',\tau')$ for slot $n_A$.
\end{lemma}
\begin{proof}
By DRBC-Reliability, if a healthy node in $\mathsf{UNL}_k^{\infty}$ broadcasts $CHECK(P,\tau)$ with $(A,n_A)\in P$, then eventually either some unblocked node in $\mathsf{UNL}_k^{\infty}$ receives strong support for $ACCEPT(A,n_A,\tau')$ for some $\tau'$ or eventually every unblocked node broadcasts $CHECK(\mathunderscore,\tau')$ for some $\tau'$ and with a $P$-set containing $(A,n_A)$. In the former case every unblocked node in $\mathsf{UNL}_k^{\infty}$ eventually adds $(A,\tau')$ as a valid input for MVBA on slot $n_A$ by Assumed-Reliability; in the latter case the same is clearly true.

Since honest nodes stop suggesting new amendments with slot number $n_A$ after they accept their first amendment through DRBC for an amendment with slot number $n_A$, if eventually every unblocked node in $\mathsf{UNL}_k^{\infty}$ accepts a valid input for slot number $n_A$, then every unblocked node can only accept a finite number of valid inputs for slot number $n$; indeed, an unblocked node in $\mathsf{UNL}_k^{\infty}$ can only accept a valid input if it some unblocked node in $\mathsf{UNL}_k^{\infty}$ suggested it, but since a node clearly cannot suggest an infinite number of amendments in a finite amount of time, only a finite number of amendments with slot number $n_A$ are supported by any unblocked node in $\mathsf{UNL}_k^{\infty}$.

Thus eventually every unblocked node in $\mathsf{UNL}_k^{\infty}$ eventually sees a common value $(A,\tau')$ as a valid input for MVBA on slot $n_A$, and the number of valid inputs for any unblocked node in $\mathsf{UNL}_k^{\infty}$ is bounded. Thus by MVBA-Termination, every unblocked node in $\mathsf{UNL}_k^{\infty}$ terminates MVBA with probability $1$.
\end{proof}

\begin{proposition}\label{propFKNonBlock}
If $\mathcal{P}_k$ is strongly connected and unblocked, and runs the waiting protocol for any time $\tau$, then eventually the waiting protocol terminates.
\end{proposition}
\begin{proof}
	Once $\mathcal{P}_i$ has received all of the $CHECK(\mathunderscore,t')$ messages from every unblocked node in $\mathsf{UNL}_k$ for every $t'\leqslant t$, then for any amendment $A$ included in one of these $CHECK$ messages, by lemma~\ref{lemmaFKDRBCAgreement} eventually every unblocked node in $\mathsf{UNL}_k^{\infty}$ accepts modified DRBC for some pair $(A,\mathunderscore)$. Thus by lemma~\ref{propDABCTermination}, $\mathcal{P}_i$ eventually ratifies some amendment for slot $n_A$.
\end{proof}

\begin{proposition}\label{propFKSufWait}
	If $\mathcal{P}_k$ is healthy and weakly connected and eventually ratifies some amendment $A$ with activation time $t$, then $\mathcal{P}_k$ will wait until it has ratified $A$ before completing the waiting protocol for any time $t'\geqslant t$.
\end{proposition}
\begin{proof}
	By lemma~\ref{lemmaFKSufficientWaiting}, if $\mathcal{P}_k$ eventually ratifies $A$ with activation time $t$ then $\mathcal{P}_k$ cannot receive $CHECK(\mathunderscore,t)$ from $q_S$ nodes in every essential subset $S\in\mathsf{ES}_k$ such that $A$ that is not present in any of the received $CHECK(\mathunderscore, t)$ messages. Thus in the waiting protocol for time $t'$, $\mathcal{P}_k$ will wait until it has ratified some amendment for slot $n_A$, and we ratify $A$ for slot $n_A$ by hypothesis.
\end{proof}

\begin{theorem}\label{thmFK}
The modified DABC protocol defined in section~\ref{sectionProtocol-FK} satisfies the properties of a democratic atomic broadcast algorithm in the open network model, along with the additional Full Knowledge property.
\end{theorem}
\begin{proof}
	Linearizability follows directly from MVBA-Consistency. Democracy follows immediately from MVBA-Validity and the corresponding Democracy property of DRBC.
	
	Liveness follows from DRBC-Censorship-Resilience and lemma~\ref{propDABCTermination}. Democracy follows from DRBC-Democracy and MVBA-Validity.
	
	Agreement follows because a healthy node only outputs $(A,\tau)$ if it received enough $ACCEPT(A,n_A,\tau)$ messages to guarantee that every unblocked node in $\mathsf{UNL}_k^{\infty}$ adds $(A,\tau)$ to its valid inputs for MVBA on slot $n_A$, in which case every unblocked node in $\mathsf{UNL}_k^{\infty}$ terminates MVBA with probability $1$, and must in fact output $A$ by MVBA-Consistency.
	
Full Knowledge follows from proposition~\ref{propFKNonBlock} and proposition~\ref{propFKSufWait}.
\end{proof}

\paragraph{Acknowledgements.} Thank you to Brad Chase and Stefan Thomas for providing helpful discussion and revisions, to Rome Reginelli for careful editing, and to David Schwartz for designing the original XRP Ledger consensus protocol, without which this research would never have been conducted. This work was funded by Ripple.

\bibliography{refs.bib}

\appendix

\section{Ordering Transactions}\label{sectionAppendix}

The discussion of Cobalt up until this point has been kept fairly general and detached from any specific use-case. However, Cobalt is intended to be used for XRP, which has a very specific use-case: the XRP Ledger is first and foremost a system for generating a public log of \textit{transactions}. Thus it would be somewhat strange to not discuss how Cobalt relates to transaction processing.

The primary goal of a decentralized transaction processing system is to determine which transactions did or did not occur. Since transactions are signed and universal constraints like ``an empty account cannot send payments" govern validity, if all nodes in the network can agree on a total ordering for the transactions then every node can independently ``apply" transactions in that order, generate consistent ledgers at every step, and agree on which transactions were valid by the universality of the constraints. Thus we consider a ``transaction processing" mechanism to be simply some mechanism which allows all nodes in the network to agree on the order in which transactions should be applied.

Since Cobalt is in particular a form of atomic broadcast algorithm, it can be directly applied to ordering transactions by sending transactions as amendments that are supported automatically if they're valid. For efficiency's sake it would be best to remove the activation time extension for this purpose, as it adds significant weight and there's no need to agree on activation times for transactions; instead a node can just add a block as a valid input for MVBA after accepting DRBC (or just regular RBC) for it. Even with the removal of activation times, this would be horribly inefficient though, since only a single transaction is accepted per MVBA instance. Further, a client with very fast network connections could censor other clients' transactions by submitting their transaction for every slot first.

An alternative is to use the ``blockchain model" and batch transactions into blocks and submit the \textit{blocks} as amendments. This is much less inefficient, but still less than optimal: if $P$ is the number of proposers and $D$ is the sum of $n_i$ across all nodes $\mathcal{P}_i$, the latency per block would likely be at least several seconds and grow logarithmically with $P$ (see appendix~\ref{sectionAppendix2}), while the communication complexity would be $O(D\cdot P)$ -- which is probably $O(n^3)$ asymptotically -- placing a relatively low limit on the possible throughput. Nonetheless, as described at the end of this section, this mechanism is effective enough to be used as a backup in emergencies, and has the benefit of being fully asynchronous unlike the alternative we present.

For these reasons, rather than having every node in the decentralized network participate in the agreement protocol for deciding the order of transactions, we recommend instead using Cobalt to vote on a universally agreed-upon set of nodes that run some fast and robust complete-network consensus algorithm like Honeybadger \citep{Miller2016} or Aardvark \citep{Clement2009} to decide on the order of transactions. In the sequel, to avoid confusion we refer to the network of nodes running Cobalt as the \textbf{Cobalt network}, and the network of nodes agreeing on transactions the \textbf{transaction network}. Changes to the transaction network are agreed upon as amendments by the Cobalt network. To ensure that nodes in the transaction network know about amendments by their activation time, we assume that nodes in the transaction network are also nodes in the Cobalt network, so that every correct node in the transaction network can reap the benefits of the full knowledge property of Cobalt. We assume that Cobalt nodes still individually validate transactions they receive from the transaction network, and throw out any transactions that are invalid, so that a malicious transaction network cannot arbitrarily modify the ledger state in illegal ways.

Clearly there is no way to guarantee forward progress if every node in the transaction network fails. However, we would like to at the very least ensure that every correct node in the Cobalt network agrees on the transaction log whenever the Cobalt network is safe, regardless of how many transaction nodes fail. To make this work, rather than simply blindly accepting blocks from the transaction network, we run a PBFT-like protocol that uses the transaction network as a distributed ``leader" and guarantees consistency even when the leader fails.

We assume that there is an infinite sequence of transaction networks (possibly not all disjoint, or possibly not even unique) which we denote by $v_1,v_2,...$ in analogy with the ``views" of PBFT. In practice Cobalt is used to agree on the sequence of views in a lazy way: amendments are proposed to add new views that can be switched to in the event that the current transaction network seems to be failing. Theoretically the views could be agreed upon in real time so that $v_{n+1}$ is decided upon only after $v_{n}$ is observed to be failing. However, designating several ``backups" in advance greatly increases the resilience and adaptability of the algorithm so that almost all issues can be detected using automated metrics and resolved in a matter of seconds using purely machine agreement.


Let $v$ be the current view, and let $t(v)$ be the threshold of tolerated faulty nodes in $v$. Further let $\mathsf{lock}(v)$ be a boolean variable for each view that initializes as false, and let $\mathsf{min}(v)$ be a positive integer constant  (in the first view of all time, $\mathsf{min}(v)=0$; for other views, $\mathsf{min}(v')$ gets set as part of the view change protocol further below).

Blocks are generated by the transaction network with increasing ``sequence numbers" describing where the block is supposed to sit in the totally ordered blockchain. When the nodes in $v$ have agreed on a block $B$ with sequence $n_B$, they each broadcast $INIT(B,n_B)$ to the Cobalt network.

A node $\mathcal{P}_i$ runs the protocol below to decide when to accept blocks from the transaction network. Note the similarity to the RBC protocol.
\begin{enumerate}
	\item Do not broadcast any messages pertaining to a sequence number $n$ unless $n\geqslant \mathsf{min}(v)$ and until we have accepted a batch for every sequence $n'$ with $\mathsf{min}(v)\leqslant n'$ and $n'<n$.
	\item Upon receiving $INIT(B,n_B,v)$ from $t(v)+1$ nodes in $v$, broadcast $ECHO(B,n_B,v)$ if we have not already broadcast $ECHO(\mathunderscore, n_B,v)$.
	\item Upon receiving weak support for $ECHO(B,n_B,v)$, broadcast $ECHO(B,n_B,v)$ if we have not already broadcast $ECHO(\mathunderscore, n_B,v)$.
	\item Upon receiving strong support for $ECHO(B,n_B,v)$, broadcast $READY(B,n_B,v)$ if we have not already broadcast $READY(\mathunderscore, n_B,v)$.
	\item Upon receiving weak support for $READY(B,n_B,v)$, broadcast $READY(B,n_B,v)$ if we have not already broadcast $READY(\mathunderscore, n_B,v)$.
	\item Upon receiving strong support for $READY(B,n_B,v)$, broadcast $CHECK(B,n_B,v)$ if $\mathsf{lock}(v)$ is false and we have not already broadcast $CHECK(\mathunderscore, n_B,v)$.
	\item Upon receiving strong support for $CHECK(B,n_B,v)$, accept the batch $B$ for sequence $n_B$.
\end{enumerate}

Clearly this shares all the same properties as a normal RBC algorithm.

To ensure that during ordinary cases (when the transaction network is not critically failing) forward progress is being made, we assume that every correct Cobalt node opens a reliable authenticated channel allowing every transaction node to broadcast to it. By RBC-Non-Triviality then, as long as the transaction network is not critically failed every Cobalt node will eventually accept every transaction batch processed by the transaction network.

By the properties of RBC, if any Cobalt node accepts some batch of transactions, then every Cobalt node eventually accepts the same batch of transactions, and two Cobalt nodes never accept inconsistent batches. Thus if any correct node observes that some transaction occurred, then every other correct node will observe that transaction occurred.

Combined with the fact that Cobalt nodes individually validate all transactions, this implies that regardless of the state of the transaction network, every correct Cobalt node is consistent and does not accept any invalid transactions, so safety is reduced purely to the correct configuration of the Cobalt network. This is a significant improvement over other algorithms that elect a transaction network but which suffer from the fact that safety is weaker than the safety of the election network.

To complete the protocol specification, nodes need a way to trigger a view change and agree on what the most recently accepted batch of transactions was so that these transactions are not overwritten in the next view. Our view change protocol is somewhat different from that of PBFT due to the lack of fully expressive cryptography in our setting.

To request a view change, $\mathcal{P}_i$ runs the following protocol.

\begin{enumerate}
	\item Broadcast $CHANGE(v')$ where $v'$ is the next view.
	\item Upon receiving strong support for $CHANGE(v')$, broadcast $CONFIRM(v')$ if we have not already done so.
	\item Upon receiving weak support for $CONFIRM(v')$, broadcast $CONFIRM(v')$ if we have not already done so.
	\item Upon receiving strong support for $CONFIRM(v')$, set $\mathsf{lock}(v)$ to true and broadcast $LOCK(v',n)$, where $n$ is the highest sequence number of any batch we have accepted from $v$.
	\item Wait until, for every essential subset $S\in\mathsf{ES}_i$, we have received $LOCK(v',\mathunderscore)$ from every node in some subset $T\subseteq S$ with $|T|=q_S$, such that if we received $LOCK(v',n)$ for any $n$ and from any node in $T$, then we have received strong support for $READY(\mathunderscore,n)$. Let $n_{locked}$ be the maximum sequence number present in any of the $LOCK(v',\mathunderscore)$ messages we received from nodes in one of the $T$ sets.
	\item If $\mathcal{P}_i$ is a member of $v'$, then $\mathcal{P}_i$ runs an external validity MVBA consensus mechanism to agree on a sequence number $n_{cont}$ which is greater than $n_{locked}$ but for which we have received strong support for $READY(B,n_{cont}-1,v)$ for some batch $B$. $\mathcal{P}_i$ then broadcasts $NEWVIEW(v',n_{cont})$.
	\item Upon receiving $NEWVIEW(v',n_{cont})$ from $t(v')+1$ nodes in $v'$, if $n_{cont}$ is greater than $n_{locked}$ and we have received strong support for $READY(B,n_{cont}-1,v)$ for some batch $B$, then broadcast $ECHO(v',n_{cont})$ if we have not already broadcast $ECHO(v',\mathunderscore)$.
	\item Upon receiving weak support for $ECHO(v',n_{cont})$, broadcast $ECHO(v',n_{cont})$ if we have not already broadcast $ECHO(v',\mathunderscore)$.
	\item Upon receiving strong support for $ECHO(v',n_{cont})$, broadcast $READY(v',n_{cont})$ if we have not already broadcast $READY(v',\mathunderscore)$.
	\item Upon receiving weak support for $READY(v',n_{cont})$, broadcast $READY(v',n_{cont})$ if we have not already broadcast $READY(v',\mathunderscore)$.
	\item Upon receiving strong support for $READY(v',n_{cont})$, for every $n<n_{cont}$ wait until we've received strong support for $READY(B,n,v)$ for some batch $B$, then accept $B$ as the batch with sequence $n$. Finally, switch the view to $v'$ and set $\mathsf{min}(v')=n_{cont}$.
\end{enumerate}

We omit the proofs that the above protocol is correct. It is very similar to the proofs of Full Knowledge in section~\ref{sectionProofs-FK}. Note that nodes can request a view change again even before receiving a $NEWVIEW$ message, which is necessary in the event that the $v'$ network starts out failed. The view change protocol can be optimized slightly further, but considering that we expect it to be rarely invoked, we opt for the less optimized protocol since we feel it is clearer.

One remaining issue with the above protocol is that if all of the planned backup views fail simultaneously, then the network can be shut down for an extended period of time until human node operators can agree on a new set of transaction nodes and ratify the amendment for it. Since the Cobalt nodes cannot distinguish node failure from communication failure, this opens a path for effectively attacking the network: launch a temporary IP routing attack against the backup views that lasts just long enough to make the Cobalt nodes panic. If the attack can last for a minute or two (just long enough to run through all of the backup views) then even after the attacker stops being active, it could take hours to restore the network.

In situations like this where we run out of backup views, we thus resort to using Cobalt to order transactions; since the alternative is total network halting, the inefficiency of Cobalt is acceptable here. The Cobalt transaction blocks are run in parallel on a separate chain from the amendments, since there's no need to order them relative to each other and doing so would harm performance. Further, Cobalt is run without activation times for agreeing on transaction blocks, since there's no need.

As it stands, Cobalt is not at all censorship resilient: a well-connected malicious node can always force its own blocks to be the ones included. We thus need to make one more small change to prevent censorship. Rather than including the slot number as part of the information in a transaction block proposal, each block is acceptable anywhere in the chain. Once a node sees a certain block $B$ as a valid input, it continues considering it as valid for all future slots, and it refuses to support any other blocks even for future slots until $B$ is ratified for some slot. This guarantees that every single block proposed will eventually be included in the chain, which trivially prevents censorship. Unlike amendments, there is no danger in allowing blocks to be placed at an indeterministic location in the chain, since the validity of each transaction can be checked externally. However, the performance is clearly very poor when the blocks have high overlap, which is why we refrain from using this mechanism in the ordinary case.

\section{Implementing Cryptographic Randomness}\label{sectionAppendix3}

In section~\ref{sectionProtocol-Coin} we defined the properties of a common random source protocol. Here we describe how such a protocol can be implemented in the open network model.

To begin, suppose there is some value $s$ that can only be constructed by the adversary with negligible probability. For a given probability space $\mathcal{S}$, let $G$ be some cryptographic pseudorandom generator that is modeled as a random oracle that samples $\mathcal{S}$ \citep{Bellare1993}. Then by definition of a random oracle, $G(s)$ is a true random value until the adversary can construct $s$, which we assumed can only occur with negligible probability.

Cachin et al. construct a CRS in the complete network model by a reduction to a robust $(t+1,n)$-threshold signature scheme \citep{Cachin2005}. A robust $(t+1,n)$-threshold signature scheme is a protocol where a group of $n$ nodes has ``shares" of some secret key $s$, and can collaborate to produce a signature $\sigma(M)$ over a given message $M$ using $s$. We require that if all the unblocked nodes in the group try to sign a given message then they can eventually produce the signature, and further a computationally bounded adversary controlling up to $t$ nodes in the group with overwhelming probability cannot construct $\sigma(M)$ until at least one honest node in the group has tried to sign $M$. Thus if $M$ is a proactively agreed upon unique tag for the CRS instance, then letting the output of CRS be $G(\sigma(M))$ immediately gives a protocol that satisfies the required properties.

It is not immediately clear how to adapt this scheme to the essential subset model, where the notion of a ``threshold" is undefined. Our adaptation centers around taking a single secret $s$ and distributing it as a threshold secret among $S$ for \textit{multiple} essential subsets $S$. Thus any single such subset can reconstruct $s$ on its own. A naive implementation of this would be insecure though, since a single poorly configured essential subset could leak the secret. Ideally, the only assumption that $\mathcal{P}_i$ should need to make is that the essential subsets in $\mathsf{ES}_i$ are all well-configured, since otherwise $\mathcal{P}_i$ can't guarantee termination regardless.

To enable every node to verify locally that the secret cannot be leaked to the adversary, we suppose informally that there exists a way of combining several values such that if any single value is secret then the output is also secret. For example, concatenating the values and running them through a random oracle would suffice. We call such a function a mixer.

Now suppose $\mathcal{P}_i$ has some secret $s$ with a corresponding public key $p$. We use an asynchronous verifiable secret sharing (AVSS) scheme. An AVSS protocol allows a specified dealer to distribute shares of a secret $s$ between a set of nodes in a way that an honest node which terminates can guarantee with overwhelming probability that shares of the actual secret corresponding to $p$ has been distributed to all the honest nodes in the group, even if the dealer is Byzantine. For example, the scheme presented by Cachin et al. would work without modification \citep{Cachin2002}. Using such an AVSS scheme, $\mathcal{P}_i$ can distribute $(t_S+1,n_S)$-threshold shares of $s$ to every essential subset $S\in\mathsf{ES}_i$. As mentioned in section~\ref{sectionModel}, $\mathcal{P}_i$ may have to pay a fee or provide a proof-of-work in order to convince the nodes in these sets to participate in its secret sharing protocols, but we assume that if $\mathcal{P}_i$ is non-faulty and reasonably determined then it can successfully distribute $s$.

Although the same $s$ is distributed to each essential subset, we assume that for any two essential subsets $S,S'\in\mathsf{ES}_i$, and any two subsets $T\subseteq S,T'\subseteq S'$ with $|T|\leqslant t_S,|T'|\leqslant t_{S'}$, the shares of $s$ in $T$ are independent of the shares of $s$ in $T'$. This can be achieved for example with Shamir's threshold sharing scheme \citep{Shamir1979} by generating a different polynomial $p_S(x)=s+c_{S,1}x+...+c_{S,{t_S}}x^{t_S}$ for each essential subset $S\in\mathsf{ES}_i$, where the non-$s$ coefficients are all uniformly sampled and independent between essential subsets.

We introduce the notion of a pseudo-amendment as an amendment that doesn't have an actual ``proposer". Instead, some external mechanism allows nodes to learn about the amendment details, and then they support it as usual by broadcasting an $ECHO$ message for it. After determining that AVSS succeeded, a node $\mathcal{P}_j$ in one of $\mathcal{P}_i$'s essential subsets broadcasts a confirmation $ALLOW(p)$ where $p$ is the public key corresponding to $s$. If a node receives weak support for $ALLOW(p)$, then it votes to support a Cobalt pseudo-amendment that adds $p$ to a common set of ``randomizing keys". Thus honest, weakly connected nodes are guaranteed to have their randomizing key accepted by DABC-Liveness (since adding randomization keys does not contradict any other amendments, if a slot fails to add $p$ then nodes can try again; we assume that the technique mentioned at the end of appendix~\ref{sectionAppendix} for guaranteeing full Censorship-Resilience is used so that $p$ is eventually accepted).

The general idea is to create signatures over a message $M$ corresponding to each randomization key, and then mix them all together to create the seed for the random function $G$. By the definition of mixing, adding an extra randomizing key cannot decrease the security of the overall protocol, since as long as the secret a single randomizing key is secure then the result of mixing signatures is also secure.

CRS-Agreement follows immediately from the DABC-Agreement and DABC-Full-Knowledge properties of Cobalt. Indeed, for any given time $\tau$, every node agrees on the set of amendments activated before $\tau$, so every node agrees on the same set of randomizing keys. Since any node can verify a signature locally, every node that outputs a signature over the specified tag $M$ for every randomizing key must output the exact same set of signatures, and thus produces the same result for CRS.

CRS-Termination follows by DABC-Democracy and the assumed robustness of the threshold signature scheme. Because of the way we use $ALLOW$ messages, DABC-Democracy only guarantees that for any weakly connected unblocked node $\mathcal{P}_i$ and any randomizing key $p$, there is some unblocked node in $\mathsf{UNL}_i$ that can receive shares of the signature corresponding to $p$ from one of \textit{its} essential subsets. Thus we assume that nodes that receive shares of $\sigma(M)$ echo the message after they have successfully reconstructed it. Since $\mathcal{P}_i$ can verify the authenticity of $\sigma(M)$ locally, this does not hamper safety and allows $\mathcal{P}_i$ to eventually produce an output.

CRS-Randomness is simply by reduction to the security of the threshold signature scheme. We can assume that $\mathcal{P}_i$ has distributed its secret and successfully planted a public key $p$ among the randomization keys (which requires only that $\mathcal{P}_i$ was at one point correct and weakly connected). Then by the definition of mixing, the output of CRS cannot be predicted until the signature over $M$ corresponding to $p$ is known. By threshold security and our assumptions about $\mathsf{ES}_i$, this cannot occur until some honest node in one of $\mathcal{P}_i$'s essential subsets has revealed its signature share over $M$ corresponding to $p$. Thus by modeling $G$ as a random oracle, we have that with overwhelming probability the adversary cannot distinguish in advance a true random variable sampled over $\mathcal{S}$ from the output of CRS, since the output of CRS is by definition $G$ applied to the mixed signatures.

An unfortunate requirement of this system is that it requires consensus to be running properly for new nodes to add their own randomization keys. Thus if the adversary is ever able to compromise every single randomization key, then theoretically the system may be unable to ever recover. It is unclear if it is possible to construct an efficient CRS system in our network model that is capable of recovering from total compromise. Nonetheless, in practice this is unlikely to be an issue: assuming a decent initial setup, the likelihood of every randomization key ever being simultaneously compromised is very low, and even with foresight of the CRS output values, in practice it would be very difficult for the adversary to prevent termination of Cobalt for an extended period of time, so recovery even from total compromise should always be possible in practice.

\section{Logarithmic Time MVBA}\label{sectionAppendix2}

Although the results in section~\ref{sectionProofs-MIBA} fully prove correctness of the MVBA protocol, so far we have only shown that the number of rounds MVBA could theoretically take is bounded by the number of valid inputs, which would imply rather poor worst-case performance. The following proposition refines the performance analysis and proves that for a large enough hash function $H$, the \textit{expected} number of rounds is in fact at most logarithmic in the number of valid inputs. This shows that Cobalt is actually reasonably efficient.

\begin{proposition}\label{propDABCFastTermination}
	Suppose $H$ is a random oracle. For any strongly connected node $\mathcal{P}_k$, if $\mathcal{P}_i\in\mathsf{UNL}_k^{\infty}$ is unblocked, then if the image of $H$ is large enough the expected number of rounds after which MVBA terminates is at most $c+\log_3(|S_i^0|)+O(1/|S_i^0|)$ where $c$ is a small constant $c\approx 0$.
\end{proposition}
\begin{proof}
	To show that MVBA is expected to terminate at or before the $R$-th round, it suffices to show that the expected number of rounds until $|S_i^{r}|=0$ is at most $R+1$. We do this by showing that the random oracles force a constant fraction of possible values to be cut out each round, and then compute the expected value analytically.
	
	First, suppose $A\in S_i^0$ for any unblocked node $\mathcal{P}_i\in\mathsf{UNL}_k^{\infty}$. Then by Assumed-Validity, there must be some unblocked node $\mathcal{P}_j\in\mathsf{UNL}_k^{\infty}$ such that for every $S\in\mathsf{ES}_j$ the majority of nodes in $S$ suggested $A$ before beginning MVBA. If $\mathcal{P}_{i'}\in\mathsf{UNL}_i$ is healthy and sampled $\rho_r$ for any $r\geqslant 0$, then because $\mathcal{P}_{i'}$ waits for enough $CONT$ messages (which can only be sent by nodes that have started MVBA) before sampling $\rho_r$, strong connectivity implies that some honest node in $\mathsf{UNL}_k^{\infty}$ must have begun MVBA before $\mathcal{P}_{i'}$ sampled $\rho_r$ and also suggested $A$ before beginning MVBA. Thus $A$ must have been chosen before $\mathcal{P}_{i'}$ sampled $\rho_r$. By CRS-Randomness, if $s_r$ is the value returned by $\rho_r$, the probability of the adversary being able to construct $s_r$ at the time of choosing $A$ is negligible. Thus with overwhelming probability, given any $A,A'\in S_i^0$ and $r,r'\geqslant 0$ with $A\neq A'$ and/or $r\neq r'$, $\mathcal{I}_r(A)$ and $\mathcal{I}_{r'}(A')$ are independent uniform random variables sampled from the image of the hash functions.
	
	For any healthy node $\mathcal{P}_i\in\mathsf{UNL}_k^{\infty}$ that gets past step~\ref{MVBALongStep} in round $r$, let $C_i\subseteq\mathsf{values}_i^{r}$ be the set for which $\mathcal{P}_i$ saw strong support for $CONT(C_i,r)$. If $\mathcal{P}_i,\mathcal{P}_j\in\mathsf{UNL}_k^{\infty}$ are both healthy and get past step~\ref{MVBALongStep} in round $r$, then, since $\mathcal{P}_k$ is strongly connected by assumption, $\mathcal{P}_i$ and $\mathcal{P}_j$ are fully linked, so some honest node must have sent both $CONT(C_i,r)$ and $CONT(C_j,r)$. But honest nodes only send $CONT$ messages for increasing subsets, so either $C_i\subseteq C_j$ or $C_j\subseteq C_i$.
	
	Thus by transitivity of set inclusion, there exists some healthy node $\mathcal{P}_i\in\mathsf{UNL}_k^{\infty}$ such that $C_i\subseteq C_j$ for every other healthy node $\mathcal{P}_j\in\mathsf{UNL}_k^{\infty}$. In particular, there exist at least two values $A_1,A_2$ such that for every healthy node $\mathcal{P}_j\in\mathsf{UNL}_k^{\infty}$, $A_1$ and $A_2$ are contained in $\mathsf{values}_j^{r}$ before $\mathcal{P}_j$ samples $\rho_{r}$. Let $L$ be the size of the image of $H$, and, for simplicity of notation, suppose without loss of generality that the image of $H$ is $\{0,...,L-1\}$.
	
	Let $x_{r}=\min\{\mathcal{I}_{r}(A_1),\mathcal{I}_{r}(A_2)\}$. Since $A_1$ and $A_2$ are both in $\mathsf{values}_i^{r}$ before any healthy node in $\mathsf{UNL}_k^{\infty}$ queries the oracle $\rho_{r}$, this guarantees that $\mathcal{I}_{r}(\mathsf{est}_i^{r+1})\leqslant x_{r}$ by the mechanism for selecting $\mathsf{est}_i^{r+1}$ in step~\ref{MVBALongStep}. Since a healthy node $\mathcal{P}_i\in\mathsf{UNL}_k^{\infty}$ only broadcasts $INIT(A,r+1)$ if $\mathcal{I}_{r}(A)\leqslant \mathcal{I}_{r}(\mathsf{est}_j^{r+1})$ for some healthy node $\mathcal{P}_j\in\mathsf{UNL}_i^{\infty}\subseteq \mathsf{UNL}_k^{\infty}$, this guarantees that if any healthy node $\mathcal{P}_i\in\mathsf{UNL}_k^{\infty}$ adds $A$ to $\mathsf{values}_i^{r+1}$, then $\mathcal{I}_{r}(A)\leqslant x_{r}$. Since $A_1$ and $A_2$ are both in $\mathsf{values}_i^{r}$ before any healthy node in $\mathsf{UNL}_k^{\infty}$ queries the oracle $\rho_{r}$, these indices are independent uniform random variables with overwhelming probability. Therefore a simple computation gives us $\mathrm{Pr}[x_{r}=k]=(2L-2k-1)/L^2|+\epsilon(k)$ where $|\epsilon(k)|$ is negligible for every $k\in\{0,...,L-1\}$.
	
	Let $P_{r}$ be the probability that a given value $A\in S_i^0$ is also in $S_i^{r}$. Since a value $A$ is in $S_i^{r}$ only if $H_{r'}(A)\leqslant x_{r'}$ for every $r'<r$, the probability that a given value in $S_i^0$ is in $S_i^{r}$ is at most
	\begin{align*}
	P_{r}&\leqslant \mathrm{Pr}[H_{0}(A)\leqslant x_{0},H_{1}(A)\leqslant x_{1},...,H_{r-1}(A)\leqslant x_{r-1}]\\
	&=\prod_{i=0}^{r-1} \mathrm{Pr}[H_{i}(A)\leqslant x_{i}].
	\end{align*}
	
	Partitioning the sample space and summing over all possible values of $x_i$ gives
	
	\begin{align*}
	\mathrm{Pr}[H_{i}(A)\leqslant x_{i}]&=\sum_{k=0}^{L-1}\frac{k+1}{L}\cdot \mathrm{Pr}\left[x_{i}=k\right]\\
	&=\sum_{k=0}^{L-1} \frac{(k+1)(2L-2k-1)}{L^3}+\frac{(k+1)\epsilon(k)}{L}\\
	&\leqslant \frac{1}{3} + \frac{7}{6L} + \frac{1}{L^3}+\epsilon
	\end{align*}
	for negligible $\epsilon$. Define $q=\frac{1}{3} + \frac{7}{6L} + \frac{1}{L^3}+\epsilon$. Thus $P_{r}\leqslant q^{r}$ for all $r\geqslant 0$. We can model this as a game where we start with $|S_i^0|$ balls and proceed to throw them into an urn with a $q$ chance of each ball landing in the urn. We discard any balls that fall out of the urn and repeat this process until the urn is empty, and ask for the expected number of rounds this takes. This problem is investigated by Szpankowski and Vernon \citep{Szpankowski1990} who prove the expected round $R$ after which the urn empties is
	
	\begin{align*}
	E[R]=\frac{\ln\left(|S_i^0|\right)+\gamma}{-\ln \left(q\right)}+\frac{1}{2}+\varepsilon_q+O\left(\frac{1}{|S_i^0|}\right),
	\end{align*}
	
	where $\gamma\approx 0.577$ is Euler's gamma constant and $\varepsilon_q$ is a very small value, experimentally found to be $\varepsilon_q < 3\cdot 10^{-4}$ for $q\approx 1/3$. Finally, expanding $q$ around $L=\infty$ gives
	
	\begin{align*}
	E[R]&\leqslant \frac{\ln\left(|S_i^0|\right)+\gamma}{\ln \left(3\right)}+\frac{1}{2}+\varepsilon_q+O\left(\frac{1}{|S_i^0|}\right)+O\left(\frac{1}{L}\right)+\epsilon\\
	&< \log_3\left(|S_i^0|\right)+1.03+O\left(\frac{1}{|S_i^0|}\right)+O\left(\frac{1}{L}\right)+\epsilon.
	\end{align*}
	
	The proposition follows by subtracting $1$ from $R$ to get the expected value of the number of rounds $r$ for which $|S_i^r|$ is nonempty.
\end{proof}

\end{document}